\documentclass[a4paper,UKenglish]{lipics-v2018}

\pdfoutput=1 

\usepackage{microtype}
\usepackage{enumitem}

\newcommand{\red}[1]{{\color{black}{#1}}}
\usepackage{scrextend}
\newcommand{\comment}[1]{}
\newcommand{\sv}{{\mathcal V}}
\newcommand{\se}{{\mathcal E}}
\newcommand{\scripte}{\mathcal{E}}
\newcommand{\scriptv}{\mathcal{V}}

\newcommand{\point}[3]{#1\stackrel{#3}{\Rightarrow}{#2}}
\newcommand{\notpoint}[3]{#1\stackrel{#3}{\not\Rightarrow}{#2}}

\newenvironment{proofSketch}{\noindent{\bf Proof Sketch:}}{\hspace*{\fill}\(\Box\)}

\newcommand{\sa}{{\mathcal A}}

\newtheorem{claim}{Claim}

\newcommand{\calg}{{\mathcal G}}
\newcommand{\calv}{{\mathcal V}}
\newcommand{\cale}{{\mathcal E}}

\title{Effects of Topology Knowledge and Relay Depth on Asynchronous Consensus}

\titlerunning{Topology Knowledge, Relay Depth, and Asynchronous Consensus}

\author{Dimitris Sakavalas}{Boston College, USA}{dimitris.sakavalas@bc.edu}{}{}

\author{Lewis Tseng}{Boston College, USA}{lewis.tseng@bc.edu}{}{}

\author{Nitin H. Vaidya}{University of Illinois at Urbana-Champaign, USA}{nhv@illinois.edu}{}{}

\authorrunning{D. Sakavalas, L. Tseng, and N.\,H. Vaidya}

\Copyright{Dimitris Sakavalas, Lewis Tseng, and Nitin H. Vaidya}

\subjclass{\ccsdesc[100]{ CCS~Computer systems organization~Dependable and fault-tolerant systems and networks~  Fault-tolerant network topologies }}

\keywords{Asynchronous systems, crash fault, consensus, incomplete graphs, topology knowledge}

\category{}

\relatedversion{}

\supplement{}

\funding{}

\acknowledgements{}

\EventEditors{John Q. Open and Joan R. Access}
\EventNoEds{2}
\EventLongTitle{42nd Conference on Very Important Topics (CVIT 2016)}
\EventShortTitle{CVIT 2016}
\EventAcronym{CVIT}
\EventYear{2016}
\EventDate{December 24--27, 2016}
\EventLocation{Little Whinging, United Kingdom}
\EventLogo{}
\SeriesVolume{42}
\ArticleNo{23}
\nolinenumbers 
\hideLIPIcs  

\begin{document}

\maketitle


\begin{abstract}

\normalsize
Consider a point-to-point message-passing network.
We are interested in the \textit{asynchronous} crash-tolerant consensus problem in \textit{incomplete networks}. 
We study the feasibility and efficiency of approximate consensus under different restrictions on topology knowledge and the \emph{relay depth},  i.e., the maximum number of hops any message can be relayed. These two constraints are common in large-scale networks, and are used to avoid memory overload and network congestion respectively.  Specifically, for different values of integers $k,k'$, we consider that each node knows all its neighbors of at most $k$-hop distance (\emph{$k$-hop topology knowledge}), and the relay depth is $k'$.
We consider both \textit{directed} and \textit{undirected} graphs. More concretely, we answer the following main question in asynchronous systems:\smallskip 
\begin{addmargin}[4em]{4em}
\textit{What is a tight condition on the underlying communication graphs for achieving approximate consensus if each node has only a $k$-hop topology knowledge and relay depth $k'$?} 	\smallskip	
\end{addmargin}
To prove that the necessary conditions presented in the paper are also sufficient, we have developed algorithms that achieve consensus in  graphs satisfying those conditions:
\begin{itemize}[noitemsep,topsep=0pt]
	\item The first class of algorithms requires $k$-hop topology knowledge and relay depth $k$. Unlike prior algorithms, these algorithms \textit{do not} flood the network, and each node \textit{does not} need the full topology knowledge. We show how the convergence time and the message complexity of those algorithms is affected by $k$, providing the respective upper bounds.
	\looseness=-1
	
	\item  The second set of algorithms requires only one-hop neighborhood knowledge, i.e., immediate incoming and outgoing neighbors, but needs to flood the network (i.e., relay depth \red{is} $n$, where $n$ is the number of nodes). One result that may be of independent interest is a \textit{topology discovery} mechanism to learn and ``estimate'' the topology in asynchronous directed networks with crash faults.
\end{itemize}
 \end{abstract}

\section{Introduction}
\label{s_intro}

The fault-tolerant consensus problem proposed by Lamport et al. \cite{lamport_agreement} has been studied extensively under different point-to-point network models, including complete networks (e.g., \cite{lamport_agreement,AA_Dolev_1986,abraham_04_3t+1_async}) and undirected networks (e.g., \cite{impossible_proof_lynch,dolev_82_BG}). 
Recently, many works are exploring various consensus problems in directed networks, e.g.,  \cite{approximate_consensus_dynamic,k-set_dynamic,k-set_dynamic_NETYS,reliable_comm_dynamic,Ashish_DISC17}, including our own work \cite{Tseng_podc2015,vaidya_PODC12,lili_multihop}. More precisely, these works address the problem in {\em incomplete} {\em directed} networks, i.e., not every pair of nodes is connected by a channel, and the channels are not necessarily bi-directional.
We will often use the terms {\em graph} and {\em network} interchangeably.
In this work, we explore the crash-tolerant approximate consensus problem in \textit{asynchronous} incomplete networks under different \red{restrictions} on \textit{topology knowledge} -- where we assume that each node knows \red{all its neighbors of at most $k$-hop distance} -- and \textit{relay depth} -- the maximum number of hops that information (or a message) can be  propagated. These constraints are common in large-scale networks \red{to avoid memory overload and network congestion}, e.g., neighbor table and Time to live (TTL) (or hop limit) in the Internet Protocol. We consider both undirected and directed graphs in this paper.

\medskip 
\noindent\textbf{Motivation}~
Prior results \cite{Tseng_podc2015,Ashish_DISC17} showed that exact crash-tolerant consensus is solvable in \textit{synchronous} networks with only one-hop knowledge and relay depth $1$, i.e., each node only needs to know its immediate incoming and outgoing neighbors, and no message needs to be relayed (or forwarded). 
Such a local algorithm is of interest in practice due to low deployment cost and low message complexity in each round.
In \textit{asynchronous} undirected networks, there exists a simple flooding-based algorithm adapted from \cite{impossible_proof_lynch,dolev_82_BG} that achieves approximate consensus with up to $f$ crash faults if the network satisfies $(f+1)$ node-connectivity\footnote{For brevity, we will simply use the term ``connectivity'' in the presentation below.} and $n > 2f$, where $n$ is the number of nodes. 
However, these two conditions are \textit{not} sufficient for an \textit{iterative} algorithm with one-hop knowledge and relay depth $1$\red{, in which each node maintains a state and exchanges state values with only one-hop neighbors in each iteration.}

\begin{figure}[h]
	\centering
	\begin{subfigure}[c]{0.45\textwidth}
		\includegraphics[width=0.6\textwidth]{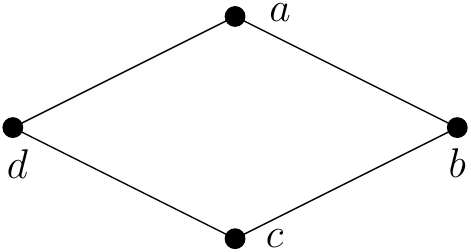}
		\caption{}
		\label{fig:introexample1}
	\end{subfigure}
	\begin{subfigure}[c]{0.45\textwidth}
		\includegraphics[width=0.7\textwidth]{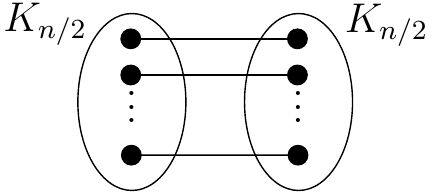}
		\caption{ }
		\label{fig:introexample2}
	\end{subfigure}\hspace{50pt}
	\caption{Effect of increased $k$-hop knowledge and relay depth $k$. In both figures, asynchronous consensus with $f=1$ is impossible for $k=1$, but possible for $k=2$.}\label{fig:introexamples}
\end{figure}

Consider Figure~\ref{fig:introexample1}, which is a ring network of four nodes. 
\red{There is no iterative algroithm with one-hop knowledge and relay depth $1$ under one crash fault}.
The adversary can divide the nodes into disjoint sets $\{a, b\}$ and $\{c, d\}$ such that the communication delay across sets is so large that $a$ thinks $d$ has crashed, and $d$ thinks $a$ has crashed, and similarly for the pair $b$ and $c$. 
As a result, no \red{exchange of state values} is possible across the sets in the execution; hence, consensus is not possible (a more precise discussion in Section \ref{s_topo}).
On the other hand, suppose each node has two-hop knowledge, i.e., a complete topology knowledge in this network, and relay depth $2$. Then $a$ knows that it will be able to receive \red{state values} from at least two of the other nodes since the node connectivity is $2$, and up to one node may fail. Following this observation, it is easy to design a flooding-based algorithm in the ring network based on \cite{impossible_proof_lynch,dolev_82_BG}.
This example shows that both topology knowledge and relay depth affect the feasibility of asynchronous approximate consensus.

Interestingly, increasing connectivity alone does \textit{not} make iterative algorithm feasible. In Section \ref{ss_fault}, we show that no fault-tolerant approximate consensus algorithm with one-hop topology and relay depth $1$ exists in the network in Figure~\ref{fig:introexample2}, which has two sparsely-connected cliques of size $n/2$ and connectivity $n/2-1$.
Motivated by these observations, this work addresses the following question in \textit{asynchronous} systems:
\begin{center}
	\fbox{
		\parbox{290pt}{
			What is a tight condition on the underlying communication graphs for achieving approximate consensus if each node has only a $k$-hop topology knowledge and relay depth $k'$?
		}
	}
\end{center}

\noindent\textbf{Approximate Consensus}~
We focus on the asynchronous approximate consensus problem. The system consists of $n$ nodes, of which at most $f$ nodes may crash. Each node is given an input, and after a finite amount of time, each fault-free node should produce an output, which satisfies \textit{validity} and {\em agreement} conditions (formally defined later).  Intuitively, the state at fault-free nodes must be in the range of all the inputs, and are guaranteed to be within $\epsilon$ of each other for some $\epsilon > 0$ after a sufficiently large number of rounds.\footnote{In the literature, it is also called asymptotic consensus. Here, we use the term ``approximate consensus'' following the work \cite{AA_Dolev_1986,Tseng_podc2015}}

In \cite{Tseng_podc2015}, we presented Condition CCA (definition in Section \ref{s_prelim}) and showed that it is necessary and sufficient on the underlying directed graphs for achieving \textit{approximate} consensus in asynchronous systems \cite{Tseng_podc2015}.
The approximate consensus algorithms in prior work \cite{Tseng_podc2015,impossible_proof_lynch,dolev_82_BG} are based on flooding (i.e., relay depth $n$) and assume that each node has $n$-hop topology knowledge.
However, such an algorithm in \textit{not} practical in a large-scale network, since nodes' local memory may not be large enough to store the entire network, flooding-based algorithms (e.g., \cite{Tseng_podc2015,impossible_proof_lynch,dolev_82_BG}) incur prohibitively high message overhead for each phase, and complete topology knowledge may require a high deployment and configuration cost.
Therefore, we explore algorithms that only require ``local'' knowledge and limited  message relay. 

\medskip
\noindent\textbf{Contributions}~We identify tight conditions on the graphs under different assumptions on topology knowledge and relay depth. Particularly, we have the following results:\smallskip
\begin{itemize}[noitemsep,topsep=0pt]	
	\item \textit{Limited Topology Knowledge and Relay Depth} (Section \ref{s_topo}): 
	We consider the case with $k$-hop topology knowledge and relay depth $k$.
	The family of algorithms that captures these constrains are iterative $k$-hop algorithms -- nodes only have topology knowledge of their $k$-hop neighborhoods, and propagate state values to nodes that are at most $k$-hops away. 
	Note that \textit{no} other information is relayed.
	For iterative $k$-hop algorithms, we derive a family of tight conditions, namely Condition $k$-CCA for $1 \leq k \leq n$, for solving \textit{approximate} consensus in directed networks. To prove the tightness of the conditions, we propose a family of iterative algorithms called $k$-LocWA and show how the convergence time and the message complexity of those algorithms is affected by $k$, providing the respective upper bounds. \looseness=-1
	
	\item \textit{Topology Discovery and Unlimited Relay Depth} (Section \ref{s_topology}): 
	We consider the case with \red{one}-hop topology knowledge and relay depth $n$.
	In other words, nodes initially only know their immediate incoming and outgoing neighbors, but nodes can flood the network, learn (some part of) the topology, and eventually solve consensus based on the learned topology.
	We show that Condition CCA from \cite{Tseng_podc2015} is also sufficient in this case. Since we assume only one-hop knowledge, our result implies that Condition CCA is tight for any $k$-hop topology knowledge. 
	One contribution that may be of independent interest is a \textit{topology discovery} mechanism to learn and ``estimate'' the topology in asynchronous directed networks with crash faults. Such a discovery mechanism will be useful for self-stabilization and reconfiguration of a large-scale system.
	
	
	\comment{*********
	\begin{itemize}
		\item While node connectivity can be linked to Condition CCA naturally, it does \textit{not} capture Condition $k$-CCA, especially when $k = 1$. This is not surprising, since node connectivity is essentially a ``global'' condition, whereas Condition $k$-CCA captures both the ``global'' and ``local'' perspectives.
		\item There exists a graph in which Condition $1$-CCA does not hold for $f \geq 1$, but Condition $2$-CCA holds for $f \leq (n - 1)/2$ for an even $n$. Note that $(n - 1)/2$ is the maximum number of tolerable faults in an asynchronous system with $n$ nodes.
		\item There exist a family of graphs and a particular delay scenario such that $1$-hop algorithm is significantly faster than $k$-hop algorithms for any $k > 1$.
	\end{itemize}
	********8}
\end{itemize}
In Section \ref{s_diss}, we discuss fault-tolerance implications of the derived conditions and Condition CCA. We also discuss how to speed up our algorithms in terms of real time delay.

\medskip
\noindent\textbf{Related Work}~
There is a large body of work on fault-tolerant consensus. Here, we discuss related works exploring consensus in different assumptions on graphs. Fisher et al. \cite{impossible_proof_lynch} and Dolev \cite{dolev_82_BG} characterized necessary and sufficient conditions under which Byzantine consensus is solvable in {\em undirected} graphs.
In synchronous systems, Charron-Bost et al. 
\cite{approximate_consensus_dynamic,approximate_consensus_dynamic_ICALP} solved {\em approximate} crash-tolerant consensus in dynamic directed networks using local averaging algorithms, and in the asynchronous setting, Charron-Bost et al.  \cite{approximate_consensus_dynamic,approximate_consensus_dynamic_ICALP} addressed approximate consensus with crash faults in {\em complete} graphs which are necessarily undirected. We solve the problem in {\em incomplete directed} graphs in asynchronous systems. Moreover, in \cite{approximate_consensus_dynamic,approximate_consensus_dynamic_ICALP}, nodes are \textit{constrained} to only have the one-hop topology knowledge. We study different types of algorithms, including the ones that allow nodes to learn the topology (i.e., we allow topology discovery).

There were also works studying limited topology knowledge. Su and Vaidya \cite{lili_multihop} identified the condition for solving synchronous Byzantine consensus using a variation of $k$-hop algorithms. Alchieri et al. \cite{BFT-CUP_OPODIS} studied the synchronous Byzantine problem under \textit{unknown} participants. We consider \textit{asynchronous} systems in this work. Nesterenko and Tixeuil \cite{topolog_discovery_06} studied the topology discovery problem in the presence of Byzantine faults in \textit{undirected} networks, whereas we present a solution that works in \textit{directed} networks with crash faults.

Extensive prior works studied graph properties for other similar problems in the presence of Byzantine failures, such as (i) Byzantine approximate consensus in directed graphs using ``local averaging'' algorithms wherein nodes only have one-hop neighborhood knowledge (e.g., \cite{vaidya_PODC12,Tseng_general,lili_multihop,Sundaram_journal,Sundaram,Sundaram_ACC,MIT17}), (ii) Byzantine consensus with unknown participants \cite{BFT-CUP_OPODIS}, (iii) Byzantine consensus with \textit{authentication} in {\em undirected} networks \cite{Bansal_disc11}.
These papers only consider synchronous systems, and our algorithms and analysis are significantly different from those developed for Byzantine algorithms, and (iv) consensus problems in \textit{synchronous} dynamic networks where the adversary can change the network topology. In this line of work, impossibility results for Consensus and $k$-Set Agreement are given in~\cite{dynamic_agreement, dynamic_k-set} and sufficiency is guaranteed by requiring a period of stability, during which certain nodes are strongly connected; the first tight condition for the feasibility of consensus and broadcast is presented in~\cite{dynamic_consensus}. Additionally, in~\cite{dynamic_Byzantine}, byzantine corruptions and a dynamic node set is assumed and a $O(\log^3n)$-round randomized algorithm is presented.   Our work is different from all these works because of the assumption of asynchronous systems and limited topology information. Please refer to our technical report \cite{crash+async_report} for further discussion on these works.

\section{Preliminary}
\label{s_prelim}

Before presenting the results, we introduce our systems model, some terminology, and our prior results from \cite{Tseng_podc2015} to facilitate the discussion.

\noindent{\bf System Model}~~
The point-to-point message-passing network is {\em static}, and it is represented by a simple {\em directed} graph $G(\scriptv,\scripte)$, where $\scriptv$ is the set of $n$ nodes, and $\scripte$ is the set of directed edges between the nodes in $\scriptv$. The communication links are reliable. We assume that $n\geq 2$, since the consensus problem for $n=1$ is trivial.  Node $i$ can transmit messages to another node $j$ directly if directed edge $(i,j)$ is in $\scripte$. Each node can send messages to itself as well; however, for convenience, we exclude self-loops from set $\scripte$.
We will use the terms {\em edge} and {\em link} interchangeably. 

Up to $f$ nodes may suffer crash failures in an execution. A node that suffers a crash failure simply stops taking steps (i.e., fail-stop model). We consider the \textit{asynchronous} message-passing communication, in which a message may be delayed arbitrarily but eventually delivered if the receiver node is fault-free. We assume that the adversary has both the control of crashing nodes and delaying messages at any point of time during the execution.

\noindent{\bf Terminology}~~
Upper case letters are used to name sets.
Lower case italic letters are used to name nodes. All paths used in our discussion are directed paths.

Node $j$ is said to be an incoming neighbor of node $i$ if $(j,i)\in \se$.
Let $N_i^-$ be the set of incoming neighbors of node $i$, i.e., $N_i^- = \{ j~|~(j, i) \in \se \}$. Define
$N_i^+$ as the set of outgoing neighbors of node $i$, i.e., $N_i^+ = \{j~|~(i,j) \in \se\}$.

For set $B\subseteq \sv$, node $i$ is said to be an incoming
neighbor of set $B$ if $i\not\in B$, and there exists $j\in B$
such that $(i,j)\in \se$. Given subsets of nodes $A$ and $B$, set $B$ is said to have $k$ incoming neighbors in set $A$ if $A$ contains $k$ distinct incoming neighbors of $B$. 
\begin{definition}
Given disjoint non-empty subsets of nodes $A$ and $B$, $\point{A}{B}{x}$ if $B$ has at least $x$ distinct incoming neighbors in $A$. When it is not true that $\point{A}{B}{x}$, we will denote that fact by $\notpoint{A}{B}{x}$. 
\end{definition}

\noindent\textbf{Approximate Consensus}~~
For the approximate consensus problem (e.g., \cite{AA_Dolev_1986,AA_nancy,Tseng_podc2015}), it is usually assumed that each node $i$ maintains a \emph{state}  $v_i$ with $v_i[p]$  denoting the state of node $i$ at the end of phase (or iteration) $p$.  The initial state of node $i$, $v_i[0]$, is equal to the initial input provided to node $i$. At the start of phase $p~(p > 0)$, the state of node $i$ is $v_i[p-1]$. 

Let $U[p]$ and $\mu[p]$ be the maximum and the minimum state at nodes that have not crashed by the end of phase $p$. Then, a \textit{correct} approximate consensus algorithm needs to satisfy the following two conditions:\smallskip 

\begin{itemize}[noitemsep,topsep=0pt]	
	\item \textit{Validity}: ~~~~~~~~$\forall p > 0, U[p] \leq U[0]$ and $\mu[p] \geq \mu[0]$; and
	\item \textit{Convergence}: ~~$\lim_{p \rightarrow 0} U[p] - \mu[p] = 0$.
\end{itemize}

\noindent Equivalently the Convergence condition can be stated as:\smallskip\\
\centerline{$\forall \epsilon>0, \text{ there exists a phase } p_\epsilon \text{ such that  for } p>p_\epsilon, U[p] - \mu[p] <\epsilon$.}\smallskip

Towards facilitating the study of the number of phases needed for convergence and the corresponding message complexity, observe that convergence with respect to a specific $\epsilon$ must be considered. Therefore we will also use the following convergence notion.\smallskip

\begin{itemize}[noitemsep,topsep=0pt]	
\item $\epsilon$-Convergence:~$\exists p_\epsilon$, $\forall p\ge p_\epsilon$, $U[p] - \mu[p] \le\epsilon$. 
\end{itemize}


\noindent{\bf Prior Result}~~
In \cite{Tseng_podc2015}, we identified necessary and sufficient conditions on the underlying communication graphs $G(\sv, \se)$ for achieving {\em crash-tolerant consensus} in directed networks.
The theorem below requires the communication graph to satisfy Condition {\bf CCA}
(Crash-Consensus-Asynchronous).
\begin{theorem}[from \cite{Tseng_podc2015}]
\label{t_cca}
Approximate crash-tolerant consensus in asynchronous systems is feasible iff for any partition $L, C, R$ of $\sv$, where $L$ and $R$ are both non-empty,\\~~~~
either  $\point{L \cup C}{R}{f+1}$ or $\point{R \cup C}{L}{f+1}$. {\bf (Condition CCA)}
\end{theorem}
\comment{++++++++++

In Appendix \ref{section:conditionsrelation}, we show that Condition CCA implies Condition $k$-CCA and Condition CCA is equivalent to Condition $n$-CCA. The Figure~\ref{fig:conditionsoverview} below the conditions, where $\calg$ is assumed to be the set of all instances-graphs $G(\calv,\cale)$ and $k,d\in \mathbb{N}$. It is easy to see through the definitions of conditions that all inclusions are strict. For completeness, we provide example graphs at the end. 

\begin{figure}[bp]
	\centering
	\includegraphics[width=0.5\linewidth]{./images/conditionsoverview}
	\caption{Relation of conditions studied in this paper. $\calg$ denotes the space of all graphs. All inclusions are strict. By Lemma~\ref{lemma:ccak_to_cca}, the last inclusion concerning CCA and $d$-CCA is strict for some $d\le n-1$.}
	\label{fig:conditionsoverview}
\end{figure}
 ++++++++++++}

\section{Limited Topology Knowledge and Relay Depth}
\label{s_topo}

In this section, we study how topology knowledge and the relay depth affect the \textit{tight} conditions on the directed communication network. 
Particularly, we consider the case with $k$-hop topology knowledge and relay depth $k$ for $1 \leq k \leq n$.
Prior works (e.g., \cite{Tseng_podc2015,impossible_proof_lynch,dolev_82_BG}) assumed that each node has $n$-hop topology knowledge and relay depth $n$. However, in large-scale networks, such an assumption may not be realistic. 
Therefore, we are interested in the algorithms that only require nodes to exchange a small amount of information within local neighborhood (e.g.,~\cite{Peleg02,PPS17, PPS17FCT}). One other benefit is that the algorithms do not require flooding \cite{Tseng_podc2015} or all-to-all communication \cite{impossible_proof_lynch,dolev_82_BG} in each asynchronous phase.

We are interested in iterative $k$-hop algorithms -- nodes only have topology knowledge in their $k$-hop neighborhoods, and propagate state values to nodes that are at most $k$-hops away.
We introduce a family of conditions, namely Condition $k$-CCA for $1 \leq k \leq n$, which we prove necessary and sufficient for achieving asynchronous approximate consensus, through the use of iterative $k$-hop algorithms.
The results presented in this section also imply how \red{$k$} affects the \textit{tight} conditions on the directed networks  -- lower \red{$k$} requires higher connectivity of the underlying communication network. 

To the best of our knowledge, two prior papers \cite{BFT-CUP_OPODIS,lili_multihop} examined a similar problem --  \textit{synchronous} Byzantine consensus. In \cite{lili_multihop}, Su and Vaidya identified the condition under different relay depths. Alchieri et al. \cite{BFT-CUP_OPODIS} studied the problem under \textit{unknown} participants. The technique developed for asynchronous consensus in this section is significantly different.

\medskip
\noindent\textbf{Iterative $k$-hop Algorithms}~~
The iterative algorithms considered here have relay depth $k$ and require
each node $i$ to perform the following three steps in \textit{asynchronous} phase $t$: \smallskip

\noindent 1. \textit{Transmit}: Transmit messages of the form $(v_i[t-1], \cdot )$ to nodes that are reachable from node $i$ via at most $k$ hops away\red{, where $v_i[t-1]$ is the current state value.} If node $i$ is an intermediate node on the route of some message, then node $i$ forwards that message as instructed by the source;\\
2. \textit{Receive}: Receive messages from the nodes that can reach node $i$ via at most $k$ hops. Denote by $R_i[t]$ the set of messages that node $i$ received at phase $t$; and\\
3. \textit{Update}: Update state using a transition function $Z_i$, where $Z_i$ is a part of the specification of the algorithm, and takes as input the set $R_i[t]$. i.e., 

~~~~~~~~~~~~$v_i[t] : = Z_i(R_i[t], v_i[t-1])~~~\text{at node}~~i$

\smallskip
\comment{+++++++
\begin{enumerate}
	\item \textit{Transmit}: Transmit messages of the form $(v_i[t-1], \cdot )$ to nodes that are reachable from node $i$ via at most $k$ hops. If node $i$ is an intermediate node on the route of some message, then node $i$ forwards that message as instructed by the source;
	\item \textit{Receive}: Receive messages from the nodes that can reach node $i$ via at most $k$ hops. Denote by $R_i[t]$ the set of messages that node $i$ received at phase $t$; and
	\item \textit{Update}: Update state using a transition function $Z_i$, where $Z_i$ is a part of the specification of the algorithm, and takes as input the set $R_i[t]$. i.e., 
	
	~~~~~~$v_i[t] : = Z_i(R_i[t], v_i[t-1])~~~\text{at node}~~i$
\end{enumerate}
+++++++}

Note that (i) no exchange of topology information takes place in this class of algorithms, and (ii) each node's state only propagates within its $k$-hop neighborhood. For a node $i$, its \emph{$k$-hop incoming neighbors} are defined as the nodes $j$ which are connected  to $i$ by a directed path in $G$ that has $\leq k$ hops. The notion of $k$-hop outgoing neighbors is defined similarly.

\medskip
\noindent\textbf{Technique}~
The algorithms presented in this section are motivated by prior work \cite{AA_Dolev_1986,lili_multihop} including our own work \cite{Tseng_podc2015}. The algorithms are iterative and simple; thus, the proof structure shares some similarity with prior work \cite{AA_Dolev_1986,Tseng_podc2015,vaidya_PODC12}. 

Generally speaking, the proof proceeds as following: (i) nodes are divided into two disjoint sets, say $L$ and $R$ so that nodes have ``closer'' state values in each set; (ii) because each node receives an adequate set of messages, we show that under any delay and crash scenarios, at least one non-crashed node in either $L$ or $R$ will receive one message from the other set of nodes in each phase; and (iii) after enough phases, the value of all non-crashed nodes in either $L$ or $R$ will move ``closer'' to the values in the other set.
Two key novelties are: identifying the ``adequate set'' of messages that needs to be received before updating local state in each asynchronous phase, and  showing that with limited $k$-hop \red{propagation}, some node is still able to receive messages from the other set (in step (ii) above).

\comment{+++++++
Steps (i) and (iii) are inspired and adapted from \cite{AA_Dolev_1986,vaidya_PODC12}, whereas Step (ii) is inspired by \cite{Tseng_podc2015}. Two key challenges are integrating the proofs together, and showing that with limited $k$-hop information, some node is still able to receive messages from the other set. In Section \ref{ss:kCCA=1}, we discuss the simper case when $k=1$ to show the proof structure and integration. In Section \ref{ss:kCCA}, we demonstrate why it is \textit{not} as straightforward in the general $k$ case, and how we tackle Step (ii) above.
++++++++=}

\subsection{$k = 1$ Case}
\label{ss:kCCA=1}

To initiate the study, we first consider the one-hop case, where each node only knows its one-hop incoming and outgoing neighbors. The following notion is crucial for the characterization of graphs in which asynchronous approximate consensus is feasible with relay depth $1$.

\begin{definition}[$A\rightarrow B$]Given disjoint non-empty subsets of nodes $A$ and $B$, we will use the notation $A \rightarrow B$ if there exists a node $i$ in $B$ such that $i$ has at least $f+1$ distinct incoming neighbors in $A$. When it is not true that $A \rightarrow B$, we will denote that fact by $A \not\rightarrow B$. 
\end{definition}

 Condition $1$-CCA, presented below proves to be necessary and sufficient for achieving asynchronous approximate consensus with relay depth $1$.
 
\begin{definition}[Condition $1$-CCA]
For any partition $L, C, R$ of $\sv$, where $L$ and $R$ are both non-empty,
either  $L \cup C \rightarrow R$ or $R \cup C \rightarrow L$.
\end{definition}

The necessity of Condition $1$-CCA is similar to the necessity proof of Condition CCA in \cite{Tseng_podc2015} and is presented in Appendix~\ref{app_1-CCA}.
For sufficiency, we present Algorithm LocWA (Local-Wait-Average) below, which is inspired by Algorithm WA \cite{Tseng_podc2015}, and utilizes only one-hop information. Recall that by definition, no message relay with depth greater than $1$ is allowed. 
In Algorithm LocWA, $heard_i[p]$ is the set of one-hop incoming neighbors of $i$ from which $i$ has received values during phase $p$.  
\red{Each node $i$ performs the averaging operation to update its state value when Condition 1-WAIT below holds for the first time in phase $p$.}
\smallskip 


\noindent\textbf{Condition 1-WAIT}: The condition is satisfied at node $i$, in phase $p$,  when  $|heard_i[p]|\ge |N_i^-|-f$, i.e., when $i$ has not received values from a set of at most $f$ incoming neighbors.

~

\hrule

\vspace*{2pt}

\noindent {\bf Algorithm LocWA} for node $i\in \sv$

\vspace*{2pt}

\hrule

~

	$v_i[0]:=$ input at node $i$

	For phase $p \geq 1$:

		~~*On entering phase $p$:
		
		~~~~~~~$R_i[p] := \{v_i[p-1]\}$
		
		~~~~~~~$heard_i[p] := \{i\}$
		
		~~~~~~~Send message $( v_i[p-1], i, p)$ to all the outgoing neighbors
		
		\vspace*{4pt}
		
		~~*When message $(h, j, p)$ is received for the \textit{first time}:
		
		~~~~~~~$R_i[p] := R_i[p] \cup \{h\}$\vspace{1mm}  \hspace*{0.5in} // $R_i[p]$ is a multiset
		
		~~~~~~~$heard_i[p] := heard_i[p] \cup \{j\}$\vspace{1mm}            
		
		\vspace*{4pt}
		
		~~*When Condition {\em 1-WAIT} holds for the first time in phase $p$:
		\begin{equation}
		\label{mean1}
		~~~~~~~~~~v_i[p] := \frac{\sum_{v \in R_i[p]} v}{|R_i[p]|}
		\end{equation}
		~~~~~~~Enter phase $p+1$

\hrule

\vspace*{6pt}


To prove the correctness of LocWA, we will use the supplementary definitions below.

\begin{definition}
For disjoint sets $A,B$, 
$in(A \rightarrow B)$ denotes the set of
all the nodes in $B$ that each have at least $f+1$ incoming edges from
nodes in $A$. When $A\not\rightarrow B$, define $in(A\rightarrow B)=\emptyset$. Formally,
$in(A\rightarrow B) = \{~v~|\,v\in B \mbox{\normalfont~and~}~f+1\leq |N_v^-\cap A|~\}$.
\end{definition}

\begin{definition}
\label{def:absorb_sequence}
For {\em non-empty disjoint} sets $A$ and $B$, set $A$ is said to {\em propagate to} set $B$ in $l$ steps, where $l>0$,
if there exist sequences of sets $A_0,A_1,A_2,\cdots,A_l$ and $B_0,B_1,B_2,\cdots,B_l$ (propagating sequences) such that
\begin{itemize}[noitemsep,topsep=0pt]
\item $A_0=A$, ~~~~ $B_0=B$, ~~~~ $A_l = A \cup B$, ~~~~ $B_l=\emptyset$, 
 ~~~~ $B_\tau \neq \emptyset$ ~for~ $\tau<l$, ~~~~~ and
\item for $0\leq \tau\leq l-1$, (i) $A_\tau\rightarrow B_\tau$; (ii) $A_{\tau+1} = A_\tau\cup in(A_\tau\rightarrow B_\tau)$; and \\ ~~~~~~~~(iii) $B_{\tau+1} = B_\tau - in(A_\tau\rightarrow B_\tau)$.
\end{itemize}
\end{definition}
Observe that $A_\tau$ and $B_\tau$ form a partition of $A\cup B$,
and for $\tau<l$, $in(A_\tau\rightarrow B_\tau)\neq \emptyset$.
We say that \underline{set $A$ propagates to set $B$} if there is a propagating sequence for some steps $l$ as defined above.
Note that the number of steps $l$ in the above definition
is upper bounded by $n-f-1$, since set $A$ must be of size at least $f+1$ for it to propagate to $B$; otherwise, $A \not\rightarrow B$. 

Now, we present two key lemmas whose proofs are presented in Appendix~\ref{app_1-CCA2}.
In the discussion below, we assume that $G$ satisfies Condition $1$-CCA.

\begin{lemma}
\label{lemma:must_absorb}
For any partition $A,B$ of $\scriptv$, where $A,B$ are both non-empty, either  $A$ propagates to $B$, or
 $B$ propagates to $A$.
\end{lemma}

The lemma below states that the interval to which the states at all the
fault-free nodes are confined shrinks after a finite number
of phases of Algorithm LocWA. Recall that $U[p]$ and $\mu[p]$ denote
the maximum and minimum states at the fault-free nodes at the end of the $p$-th phase.\looseness=-1

\begin{lemma}
	\label{lemma:bounded_value}
	Suppose that at the end of the $p$-th phase of Algorithm LocWA, $\sv$ can be partitioned into non-empty sets
	$R$ and $L$ such that (i) $R$ propagates to $L$ in $l$ steps,
	and (ii) the states of fault-free nodes
	in $R-F[p]$ are confined to an interval of length $\leq \frac{U[p]-\mu[p]}{2}$.
	Then, with Algorithm LocWA,
	\begin{eqnarray}
	U[p+l]-\mu[p+l]~\leq~\left(1-\frac{\alpha^l}{2}\right)(U[p] - \mu[p]),~~~~\text{where}~~\alpha=\displaystyle\min_{i\in \sv} \frac{1}{|N_i^-|}
	\label{e:convergence:1}
	\end{eqnarray}
\end{lemma}

Using lemma \ref{lemma:bounded_value} and simple algebra, we can prove the following Theorem. For the sake of space, we present only a proof sketch. The \red{complete} proof is deferred to Appendix \ref{app_1-CCA2}. 

\begin{theorem}
	\label{Theorem:SufficiencyLocWA}
	If $G(\sv,\se)$ satisfies Condition $1$-CCA, then Algorithm LocWA achieves both Validity and Convergence.
\end{theorem}

\begin{proofSketch}
	To prove the Convergence of LocWA, we show that given any $\epsilon>0$, there     exists $\tau$ such that    $U[t]-\mu[t] \leq \epsilon, \forall t\geq \tau$.     Consider $p$-th phase, for some $p\geq 0$. If $U[p]-\mu[p]=0$, then the algorithm has already converged; thus, we  consider only the case where $U[p]-\mu[p]>0$. In this case, we can partition $\scriptv$ into two subsets, $A$ and $B$, such that, for each fault-free node $i\in A$, $v_i[p]\in \left[\mu[p], \frac{U[p]+\mu[p]}{2}\right)$, and or each fault-free node $j\in B$, $v_j[p] \in \left[\frac{U[p]+\mu[p]}{2}, U[p]\right]$.  (Full proof in \red{\cite{crash+async_report}}
	 identifies how to partition the nodes.)
	 By Lemma~\ref{lemma:must_absorb}, we have that either $A$ propagates to set $B$ or $B$ propagates to $A$.
	In both cases above, we have found two non-empty sets $L=A$ (or $L=B$) and $R=B$ (or $L=A$) partitioning $\sv$ and satisfy the hypothesis of Lemma~\ref{lemma:bounded_value}, since $R$ propagates to $L$ and  the states of all fault-free nodes in $R$ are confined to an interval of length $\leq \frac{U[p]-\mu[p]}{2}$. The theorem is then proven by using simple algebra and the fact that the interval to which the states of all the fault-free nodes are confined shrinks after a finite number of phases.
\end{proofSketch}

\subsection{General $k$ Case}
\label{ss:kCCA}

Now, consider the case when each node only knows its $k$-hop neighbors and the relay depth is $k$. In the following, we generalize the notions presented above to the $k$-hop case.
For node $i$, denote by $N_i^-(k)$ the set of $i$'s $k$-hop incoming neighbors,
 For a set of nodes $A$, let $N_A^-$ be the set of $A$'s one-hop incoming neighbors. Formally, $N_A^- = \{ i~|~i \in \sv-A,~\text{and}~ \exists j \in A, (i,j) \in \se\}$. Next we define the relation $A\rightarrow B$ for the $k$-hop case.

\begin{definition}[$A \rightarrow_k B$]\label{def:kpropagate}
Given disjoint non-empty subsets of nodes $A$ and $B$, we will say that $A \rightarrow_k B$ holds if there exists a node $i$ in $B$ for which there exist at least $f+1$ node-disjoint paths of length at most $k$ from distinct nodes in $N_i^- \cap A$ to $i$. More formally, if $P_i^A(k)$ is the family of all sets of node-disjoint paths (with $i$ being their only common node) initiating in $A$ and ending in node $i$, $A \rightarrow_k B$ means that
$\exists i\in B,  \max\{|p| : p\in P_i^A(k)\}\ge f+1$.
\end{definition}

\begin{definition}[Condition $k$-CCA]
For any partition $L, C, R$ of $\sv$, where $L$ and $R$ are both non-empty, either  $L \cup C \rightarrow_k R$ or $R \cup C \rightarrow_k L$.
\end{definition}


The necessity of Condition $k$-CCA for achieving asynchronous approximate consensus through an iterative $k$-hop algorithm holds analogously with the one-hop case, where a set of $x$ incoming neighbors of node $i$ has to be replaced with a set of $x$ distinct nodes that reach $i$ through disjoint paths. For sufficiency, we next present a generalization of Algorithm LocWA for the $k$-hop case. There are two differences between Algorithms $k$-LocWA and LocWA: (i) nodes transmit its state to all their $k$-hop outgoing neighbors, and (ii) Algorithm $k$-LocWA relies on the \red{generalized version of Condition 1-WAIT}, presented below. \smallskip   

\noindent\textbf{Condition $k$-WAIT}: For $F_i\subseteq N_i^-(k)$, we denote with $reach_i^k(F_i)$ the set of nodes that have paths of length $l\le k$ to node $i$ in $G_{V-F_i}$. That is, the set of $k$-hop incoming neighbors of $i$ that remain connected with $i$ even when all nodes in set $F_i$ crash.   The condition is satisfied at node $i$, in phase $p$ if there exists $F_i\subseteq N_i^-(k)$ with $|F_i[p]|\le f$  such that $reach_i^k(F_i[p])\subseteq heard_i[p]$.

%

~

\hrule

\vspace*{2pt}

\noindent {\bf Algorithm $k$-LocWA} for node $i\in \sv$

\vspace*{4pt}

\hrule

\vspace*{4pt}

	$v_i[0]:=$ input at node $i$

	For phase $p \geq 1$:

		~~*On entering phase $p$:
		
		~~~~~~~$d_i[p] :=1$
		
		~~~~~~~$R_i[p] := \{v_i[p-1]\}$
		
		~~~~~~~$heard_i[p] := \{i\}$
		
		~~~~~~~Send message $( v_i[p-1], i, p)$ to nodes in $N_i^+(k)$, all $k$-hop outgoing neighbors
		\footnote{For brevity, we do not specify how the network routes the messages within the $k$-hop neighborhood -- this can be achieved by using local flooding through tagging a hop counter in each message.}

		~~* When message $(h, j, p)$ is received for the \textit{first time}:
		
		~~~~~~~$R_i[p] := R_i[p] \cup \{h\}$\vspace{1mm}  \hspace*{0.5in} // $R_i[p]$ is a multiset
		
		~~~~~~~$heard_i[p] := heard_i[p] \cup \{j\}$\vspace{1mm}

		~~* When Condition {\em $k$-WAIT} holds for the first time in phase $p$:
		
		~~~~~~~$v_i[p] := \frac{\sum_{v \in R_i[p]} v}{|R_i[p]|}$
		
		~~~~~~~Enter phase $p+1$

\hrule

\vspace*{4pt}

~

\noindent\textbf{Correctness of Algorithm $\mathbf{k}$-LocWA}~
Proving the correctness of $k$-LocWA follows a similar reasoning of the correctness of LocWA.
The key here is to identify Condition $k$-CCA and Condition $k$-WAIT so that the proof structure remains almost identical.
To adapt the arguments to the general case, one should define the analogous $in(A \rightarrow_k B)$  definition based on the general $A\rightarrow_k B$ notion.

\begin{definition}
For disjoint sets $A,B$, 
$in(A \rightarrow_k B)$ denotes the set of all the nodes $i$ in $B$ that there exist least $f+1$ incoming disjoint paths of length at most $k$ from distinct nodes in $N_i^- \cap A$ to $i$. When $A\not\rightarrow_k B$, define $in(A\rightarrow_k B)=\emptyset$. Formally, in terminology of Definition~\ref{def:kpropagate}: 
$in(A\rightarrow B)=\{ i\in B : \max\{|p| : p\in P_i^A(k)\}\ge f+1\}$

\end{definition}

%

The correctness proof of Algorithm $k$-LocWA is similar to the proof of Theorem~\ref{Theorem:SufficiencyLocWA}; remarks on the arguments' adaptations are presented in the proof sketch of the following theorem. 

\begin{theorem}
	\label{t_cca_k}
	Approximate crash-tolerant consensus in an asynchronous system using iterative $k$-hop algorithms is feasible iff $G$ satisfies Condition $k$-CCA.
\end{theorem}

\begin{proofSketch}
	Having defined the basic notion $in(A \rightarrow_k B)$, Definition~\ref{def:absorb_sequence} of the notion \emph{$A$ propagates to $B$} is the same for the $k$-hop case. 
	Intuitively, if \emph{$A$ propagates to $B$}, information will be propagated gradually from $A$ to $B$ in $l$ steps; corruption of any faulty set of $f$ nodes will \textit{not} be able to block propagation to a specific node $i$ because the definition of $in(A \rightarrow_k B)$ guarantees that $i$ will receive information from at least $f+1$ disjoint paths if it has not crashed. A difference with the original case is that for every of the $l$ steps needed to propagate from $A$ to $B$, $k$ communication steps will be required in the worst case, since information may be propagated through paths of length $k$. Lemma~\ref{lemma:bounded_value} is intuitively the same since it is based on the general propagation notion but value $\alpha$ which is defined based on the number of incoming neighbors will now be defined on the number of $k$-hop incoming neighbors, i.e., $\alpha_k=\displaystyle\min_{i\in \sv} \frac{1}{|N_i^-(k)|}$.	The main correctness proof remains essentially the same since it repeatedly makes use of the abstract propagation notion between various sets, without focusing on how the values are propagated.
\end{proofSketch}


\subsection{Condition Relation and Convergence Time Comparison}

Next, we first compare the feasibility of approximate consensus for different values of $k$ by presenting a relation among the various $k$-CCA conditions as well as their relation with Condition CCA from~\cite{Tseng_podc2015}.

\subsubsection*{Condition Relation}

We first show that  lower $k$ requires higher connectivity of the graph $G$ as stated below.

\begin{theorem}
For values $k,k'\in \mathbb{N}$ with $k\le k'$, \red{Condition $k$-CCA implies Condition $k'$-CCA}.	
\end{theorem}

\begin{proof}
	Let Condition $k$-CCA hold and assume, without loss of generality that $L\cup C\rightarrow_k R$ holds for a partition $L,C,R$. This means that there exists a node $i$ in $R$ that has at least $f+1$ incoming disjoint paths of length at most $k$ initiating from distinct nodes in $L\cup C$. Consequently, the same $f+1$ paths  will consist $i$'s incoming disjoint paths  of length at most $k'$, since $k'\ge k$, and thus, $L\cup C\rightarrow_{k'} R$ which means that $k'$-CCA holds.  
\end{proof}

We next show that Condition CCA is equivalent to Condition $n$-CCA. The proof illustrates how the locally defined Condition $k$-CCA naturally coincides with the globally defined condition CCA in the extreme case. 

\begin{theorem}
	\label{lemma:ccak_to_cca}
	Condition CCA is equivalent to Condition $n$-CCA.
\end{theorem}

\begin{proof}
 It is easy to see that Condition $n$-CCA implies Condition CCA. If Condition CCA is violated in $G$, then Condition $n$-CCA does not hold either, since $L$ and $R$ have at most $f$ one-hop incoming neighbors.

Now, we show the other direction. Assume for the sake of contradiction that Condition CCA holds but Condition $n$-CCA does not. Then, there exists a partition $L,C,R$ with $L,R\neq \emptyset$ such that $L\cup C \not\rightarrow_k R$ and $R\cup C \not\rightarrow_k L$. Since Condition CCA holds, we have that either $\point{L\cup C}{R}{f+1}$ or $\point{R\cup C}{L}{f+1}$. Now consider the case that $\point{L\cup C}{R}{f+1}$ and $\notpoint{R\cup C}{L}{f+1}$. This means that $|N_R^-|\ge f+1$ and $|N_L^-|\le f$. The case of $\notpoint{L\cup C}{R}{f+1}$ and $\point{R\cup C}{L}{f+1}$ is  symmetrical and the case of $\point{L\cup C}{R}{f+1}$ and $\point{R\cup C}{L}{f+1}$ can be proved by applying the argument below once for set $R$ and once for set $L$. 

Let $i$ be the node in $R$ with the maximum number $m$ of disjoint paths initiating from distinct nodes in $V-R$ (as implied by Definition~\ref{def:kpropagate}).  The fact $L\cup C\not\rightarrow_k R$ implies that $m\le f$. Subsequently, $|N_R^-|\ge f+1$ implies that the set $A=N_R^- -N_i^-(n)$ is non-empty (the maximal subset of $N_R^-$  which does not contain any $n$-hop incoming neighbors of $i$). Let $B=N^+_A(n)\cap R$ be the set of all the outgoing $n$-hop neighbors of all nodes $j\in A$ confined in the set $R$. By definition of $B$ and $A$, it holds that $N_i^-(n)\cap B=\emptyset$. We can now create a new partition $L'=L,C'=C\cup B,R'=R- B$ by moving $B$ from $R$ to $C$. For partition $L',C',R'$ it holds that $L',R'\neq \emptyset$ since $i\in R'$ and $L'=L$. Moreover, it holds that  (i)  $|N_{R'}^-|\le f$, since $|N_{R'}^-|=|N_{R}^- - A|$ and $A\neq \emptyset$; and (ii) $|N_L^-|\le f$ since $L=L'$. The latter points imply that  $\notpoint{R\cup C}{L}{f+1}$ and $\notpoint{L\cup C}{R}{f+1}$, which yield a contradiction to the hypothesis that Condition CCA holds. This completes the proof.
\end{proof}

\subsubsection*{Convergence Time Comparison}
We derive upper bounds on the number of \textit{asynchronous} phases needed for $\epsilon$-convergence of Algorithm $k$-LocWA and its message complexity up to this $\epsilon$-convergence point $p_{\epsilon}$.  These upper bounds are functions of values $\epsilon, k, f, n $ and $\delta=U[0]-\mu[0]$ which are naturally expected to affect the convergence time and message complexity. Moreover, since the bounds depend on $k$, it provides a way to compare the convergence time and message complexity of Algorithms $k$-LocWA for different values of $k$.  We will use the following Lemma to compute the number of phases needed for $\epsilon$-convergence of Algorithm $k$-LocWA.

\begin{lemma} \label{lemma:generalization}
For any phase $p$ of $k$-LocWA, if $U[p]-\mu[p]=0$, then there exists an integer $l(p)$, $1\le l(p) \le n-f-1$ such that, for  $\alpha_k=\displaystyle\min_{i\in \sv} \frac{1}{|N_i^-(k)|}$, the following holds,
$$
	U[p+l(p)]-\mu[p+l(p)] \leq \left( 1-\frac{\alpha_k^{l(p)}}{2}\right)(U[p] - \mu[p])
	\label{e_t}
$$

\end{lemma}

The proof of the Lemma is given in the proof of Theorem~\ref{Theorem:SufficiencyLocWA} and is based on the generalization of Lemma~\ref{lemma:bounded_value} to the $k$-hop case, which is obtained by replacing $\alpha$ with $\alpha_k=\displaystyle\min_{i\in \sv} \frac{1}{|N_i^-(k)|}$). 
Next we present the upper bound on the  convergence time of $k$-LocWA.
The Theorem can be proved by repeatedly applying Lemma~\ref{lemma:generalization} until the value $U[p]-\mu[p]$ is less than $\epsilon$.
The full proof is in \cite{crash+async_report}.

\begin{theorem}[Convergence-time complexity]\label{theorem:econvergence}
The number of phases required by Algorithm $k$-LocWA to  $\epsilon$-converge is  $\displaystyle O \left( \frac{(n-f)\log \epsilon/ \delta}{\log \left(1-\frac{\alpha_k^{n-f-1}}{2}\right)} \right)$.

\end{theorem}

\begin{proof}

The idea is to repeatedly apply Lemma~\ref{lemma:generalization} until the value $U[p]-\mu[p]$ is less than $\epsilon$.

	Observe that $\alpha_k>0$, else Condition $k$-CCA is violated. Also, $n-f-1 \geq l(p)\geq 1$ and $0<\alpha_k\leq 1$; hence, $0\leq \left( 1-\frac{\alpha_k^{l(p)}}{2}\right)<1$. We will denote $U[0]-\mu[0]$ by $\delta$ for succinctness.  
Assume wlog that $\delta>0$, and define the following sequence of phase indices:
	\begin{itemize}[noitemsep,topsep=0pt]
		\item $\tau_0 = 0$,
		\item for $i>0$, $\tau_i = \tau_{i-1} + l(\tau_{i-1})$, where $l(p)$ for any given $p$ is defined by Lemma~\ref{lemma:generalization}.
	\end{itemize}
	
\noindent By repeated application of Lemma~\ref{lemma:generalization}, we have that for $i\geq 0$,
$$
	U[\tau_i]-\mu[\tau_i] \leq \left( \prod_{j=1}^i\left( 1-\frac{\alpha_k^{\tau_j-\tau_{j-1}}}{2}\right)\right)\delta
$$

\noindent so, $\epsilon$-convergence will be achieved in phase $\tau_i$, where $\prod_{j=1}^i\left( 1-\frac{\alpha_k^{\tau_j-\tau_{j-1}}}{2}\right)\delta\le\epsilon$. Since $\tau_j-\tau_{j-1}=l(\tau_j-1) \le n-f-1$ for every $j$,  we have that,

\begin{align*} 
\prod_{j=1}^i\left( 1-\frac{\alpha_k^{\tau_j-\tau_{j-1}}}{2}\right)\delta\le\epsilon \Rightarrow
\left( 1-\frac{\alpha_k^{n-f-1}}{2}\right)^i\delta\le\epsilon 
\Rightarrow 
&i\ge \log_{\left(1-\frac{\alpha_k^{n-f-1}}{2}\right) } ~\frac{\epsilon}{\delta} \Rightarrow \\ 
 \Rightarrow & i \ge \frac{\log \epsilon/ \delta}{\log \left(1-\frac{\alpha_k^{n-f-1}}{2}\right) }
\end{align*}

By the definition of the sequence $\tau_i$ and the bound of all $l(p)$ we have that $\tau_i\leq i (n-f-1)$. Thus, the algorithm will $\epsilon$-converge by phase $\displaystyle \frac{\log \epsilon/ \delta}{\log \left(1-\frac{\alpha_k^{n-f-1}}{2}\right)}(n-f-1)$ the latest. 
\end{proof}

\noindent \textbf{Comparison of Algorithms $k$-LocWA Convergence}~~ Observe that the above bound decreases\red{, as the} maximum number of $k$-hop incoming neighbors increases, since   $\alpha_k=\displaystyle\min_{i\in \sv} \frac{1}{|N_i^-(k)|}$. Since the  maximum number of $k$-hop incoming neighbors increases with $k$ we have that for $k'\ge k$, Algorithm $k'$-LocWA $\epsilon$-converges faster than $k$-LocWA by a factor implied by the bound. 

Moreover, given the upper bound on phases for $\epsilon$-convergence of Theorem~\ref{theorem:econvergence} we can easily derive an upper bound on the message complexity of $k$-LocWA. Namely, 

\begin{theorem}[Message Complexity]
The number of messages exchanged in an execution of Algorithm $k$-LocWA until $\epsilon$-convergence is $O\left(\displaystyle \frac{(n-f)\log \epsilon/ \delta}{\log \left(1-\frac{\alpha_k^{n-f-1}}{2}\right)}kn^2\right)$
\end{theorem}

\begin{proof}
This holds because each phase of Algorithm $k$-LocWA may require $k$ communication steps for $k$-length paths to propagate values to a receiver. In the worst case, each node sends to all of its neighbors in every communication step.
\end{proof}

\section{Topology Discovery and Unlimited Relay Depth}
\label{s_topology}

In this section, we consider the case with one-hop topology knowledge and relay depth $n$.
In other words, nodes initially only know their immediate incoming and outgoing neighbors, but nodes can flood the network and learn the topology.
The study of this case is motivated by the observation that full topology knowledge at each node (e.g., \cite{Tseng_podc2015,impossible_proof_lynch,dolev_82_BG}) requires a much higher deployment and configuration cost. 
We show that Condition CCA from \cite{Tseng_podc2015} is necessary and sufficient for solving approximate consensus with one-hop neighborhood knowledge and relay depth $n$ in asynchronous directed networks. 
Compared to the iterative $k$-hop algorithms in Section \ref{s_topo}, the algorithms in this section are \textit{not} restricted in the sense that nodes can propagate any messages to all the reachable nodes. 

The necessity of Condition CCA is implied by our prior work \cite{Tseng_podc2015}.
The algorithms presented below are again inspired by Algorithm WA from \cite{Tseng_podc2015}. The main contribution is to show how each node can learn ``enough'' topology information to solve approximate consensus -- this technique may be of interests in other contexts as well. 
In the discussion below, 
we present an algorithm that works in any directed graph that satisfies Condition CCA.
\smallskip 

\noindent \textbf{Algorithm LWA}~
The idea of Algorithm LWA (Learn-Wait-Average) is to piggyback the information of incoming neighbors when propagating state values. Then, each node $i$ will locally construct an \textit{estimated} graph $G^i[p]$ in every phase $p$, and check whether \red{Condition $n$-WAIT} holds in $G^i[p]$ or not. Note that $G^i[p]$  may not equal to $G$, as node $i$ may not receive messages from some other nodes due to asynchrony or failures. 
We say \red{Condition $n$-WAIT} holds in the local estimated graph $G^i[p](\sv^i[p], \se^i[p])$ if {\bf there exists} a set $F_i[p] \subseteq \sv^i[p] -\{i\}$, where $|F_i[p]| \leq f$, such that $reach'_i(F_i[p]) \subseteq heard_i[p]$. Here, $reach'_i(F_i)$ is the set of nodes that have paths to node $i$ in the subgraph induced by the nodes in $\sv^i[p]-F_i[p]$ for $F_i[p]\subseteq \sv^i[p] -\{i\}$ and $|F_i[p]|\leq f$.

Recall that $N_i^-$ denotes the set of $i$'s one-hop incoming neighbors. Given a set of nodes $N$ and node $i$, we also use the notation $G_{N \Rightarrow i}$ to describe a directed graph consisting of nodes $N \cup \{i\}$ and set of directed edges from each node in $N$ to $i$. Formally, $G_{N \Rightarrow i} = (N \cup \{i\}, E')$, where $E' = \{(j, i)~|~j \in N\}$.

~

\hrule

\vspace*{2pt}

\noindent {\bf Algorithm LWA} for node $i\in \sv$

\vspace*{4pt}

\hrule

\vspace*{4pt}

	$v_i[0]:=$ input at node $i$
	
	$G^i[0] := G_{N_i^- \Rightarrow i}$
	
	For phase $p \geq 1$: 
	
		~~* On entering phase $p$:
		
		~~~~~~~$R_i[p] := \{v_i[p-1]\}$
		
		~~~~~~~$heard_i[p] := \{i\}$
		
		~~~~~~~Send message $( v_i[p-1], N_i^-, i, p)$ to all the outgoing neighbors
		
		\vspace*{4pt}
		
		~~* When message $(h, N, j, p)$ is received for the \textit{first time}:
		
		~~~~~~~$R_i[p] := R_i[p] \cup \{h\}$\vspace{1mm}  \hspace*{0.5in} // $R_i[p]$ is a multiset
		
		~~~~~~~$heard_i[p] := heard_i[p] \cup \{j\}$\vspace{1mm}            
		
		~~~~~~~$G^i[p] := G^i[p] \cup G_{N \Rightarrow j}$~\footnote{$G_1(\sv_1, \se_1) \cup G_2(\sv_2, \se_2) \equiv G_3(\sv_3, \se_3)$, where $\sv_3 = \sv_1 \cup \sv_2$ and $\se_3 = \se_1 \cup \se_2$. Note that this is \textit{not} a multiset, there is only one copy of each node or edge.}~\vspace{1mm}       		
		
		~~~~~~~Send message $(h, N, j, p)$ to all the outgoing neighbors\vspace{1mm}

		\vspace*{4pt}
		
		~~* When Condition {\em $n$-WAIT} holds on $G^i[p]$ for the first time in phase $p$:
		
		~~~~~~~$v_i[p] := \frac{\sum_{v \in R_i[p]} v}{|R_i[p]|}$
		
		~~~~~~~$G^i[p+1] := G_{N_i^- \Rightarrow i}$  \hspace*{0.5in} // ``Reset'' the learned graph
		
		~~~~~~~Enter phase $p+1$

\hrule

\vspace*{4pt}

~

\noindent\textbf{Correctness of Algorithm LWA}~
The key lemma to prove the correctness of Algorithm WA in \cite{Tseng_podc2015} is to show that for any pair of nodes that have not crashed in phase $p$, they must receive a state value from at least one common node. In Appendix~\ref{app_topo_dir},
we show that Algorithm LWA achieves the same property. Intuitively, if \red{Condition $n$-WAIT} does not hold in the local estimated graph $G^i[p]$, then node $i$ knows it can learn more states in phase $p$. Also, when \red{Condition $n$-WAIT} is satisfied in $G^i[p]$, there exists a scenario that node $i$ cannot receive any more information; hence, it should not wait for any more message. This is why the Algorithm LWA allows each node to learn enough state values to achieve approximate consensus. We rely on this observation to prove the correctness in \cite{crash+async_report}.





\medskip
\noindent\textbf{Undirected Graphs}~
Algorithm LWA works on undirected graphs as well; however, the \red{message size is large}, since each message needs to include the information about one's neighborhood. 
%
In Appendix~\ref{app_topo_undir}, we present an algorithm in which each node learns the topology in the first phase, and then executes an approximate consensus algorithm using the learned topology. The reasons that this trick works in undirected graphs are: (i) Condition CCA is equivalent to $(f+1)$ connectivity and $n > 2f$ in undirected graph; and (ii) for each node, there is \red{at least} one fault-free neighbor; hence, each node is able to learn the existence of every other node. 

\section{Discussion}
\label{s_diss}

In this section, we discuss interesting implications of the conditions derived in this paper. 

\subsection{Fault-tolerance}
\label{ss_fault}

In undirected graphs, $(f+1)$-connectivity and $n > 2f$ are both necessary and sufficient for solving approximate consensus in asynchronous networks with up to $f$ crash faults (implied by \cite{impossible_proof_lynch,dolev_82_BG}). It is easy to show that Condition CCA for tolerating $f$ faults is equivalent to these two conditions in undirected networks. However, this equivalence does \textit{not} hold for general $k$. For example, the network in Figure~\ref{fig:introexample1} has connectivity $2$ and four nodes, but does not satisfy Condition $1$-CCA with $f = 1$ (when $L = \{a, b\}, R = \{c, d\}, C = \emptyset$).

More interestingly, increasing the topology knowledge and relay depth by a small amount may increase the fault-tolerance tremendously. Consider the network in Figure~\ref{fig:introexample2}. Condition $1$-CCA does not hold for $f \geq 1$ (when $L=$ left clique, $R=$ right clique, and $C = \emptyset$). On the other hand, Condition $2$-CCA holds for $f \leq n/2-1$. Intuitively, this holds because each pair of nodes are at most two hops away.

\subsection{Real Time Speed Up of Algorithm k-LocWA}

In asynchronous systems, the real time communication delay is arbitrary but finite. In a formal framework, it is common to assume that execution proceeds in rounds representing real time intervals, but the nodes  do not have knowledge of the round index.  To model the worst-case real time delay in the execution of a system we can use the notion of \emph{delay scenario} which is a description of the delays, incurring on the communication through all edges of the network. 
The delivery delay of a message sent over a channel $e$ will be described by the number of rounds (amount of real time) that are needed for the delivery to be completed.  

We first compare the real time performance of Algorithms $k$-LocWA for different values of $k$ with respect to the real time delay. Specifically we show that there is a case where Algorithm LocWA terminates each phase in one round (one interval of real time), while it may take arbitrary number of rounds for Algorithm $2$-LocWA to terminate phase 1. To formalize the comparison we will use the notion of $\epsilon$-convergence time of Algorithms $k$-LocWA.
\looseness=-1

\begin{example}
	\label{example1}
	
	Consider the graph of Figure~\ref{fig:timing1}, which is a ring network plus a \textit{directed} edge $(C, B)$.
	For $f=1$,  it is easy to verify that Condition 1-CCA holds, which implies that Conditions $i$-CCA, for $i\in \{1, \ldots, n\}$ hold. Assume that the delivery of messages through \textit{directed} edges $(A,C), (C,A), (B,D), (D,B)$ is delayed by $d$ rounds while the communication in all the other edges is instant (1 round). For ease of presentation assume that no node crashes. Then, in an execution of Algorithm LocWA, it is clear that every node $i$ will finish phase $t$ in time $t$ because in each phase,  it  will receive a message from all of his neighbors $N_i^-$ except one, in one round and thus,  Condition 1-WAIT  will be satisfied.\looseness=-1
	
	On the other hand, in an execution of Algorithm $2$-LocWA, node $D$ will only receive a message from $C$ in one round, since $(C, B)$ is a directed edge, and delay on edges $(A, C)$ and $(B, D)$ is $d$. in this case, $D$ will not be able to decide before round $d$, the first round where Condition 2-WAIT will be satisfied. Specifically, for the first phase it will hold that $reach_D^2 \subseteq heard_D[1]$ only after round $d$
	 since, if $D$ considers $F_i=\{B\}$ as a possible corruption set, it has to wait for a message from $A$ which will be propagated by $C$ and setting $F_i=\{B\}$, it has to wait for a message from $B$. Consequently the first time that node $D$ can decide is round $d$ where  it will receive the rest of the values. For similar reasons, the same holds for nodes $A, C$.  Since $d$ may be an arbitrary integer, there is a delay scenario where the $\epsilon$-convergence time for Algorithm 2-LocWA, is arbitrarily larger than  the  $\epsilon$-convergence time of Algorithm LocWA.

	\begin{figure}[tbp]
		\centering
		\begin{subfigure}[b]{0.4\textwidth}
			\includegraphics[width=0.7\textwidth]{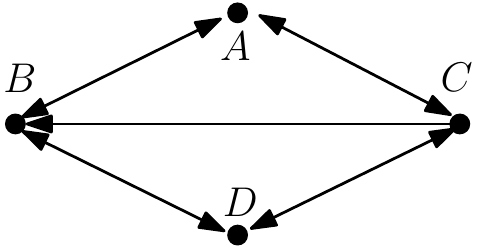}
			\caption{Graph $G$, $i$-CCA holds for any $i\in\{1, \ldots, 4\}$ and $f=1$.}
			\label{fig:timing1}
		\end{subfigure}\hspace{50pt}
		~ 
		\begin{subfigure}[b]{0.4\textwidth}
			\includegraphics[width=0.7\textwidth]{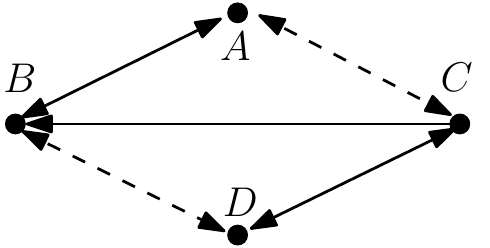}
			\caption{Arbitrary delay in directed edges $(A,C), (C,A), (B,D), (D,B)$}
			\label{fig:timing2}
		\end{subfigure}
		\caption{Real time delay example}\label{fig:timing}
	\end{figure}

\end{example}
%
%

\smallskip
\noindent\textbf{Strong version of $\mathbf{k}$-LocWA with respect to real time}~
In Example~\ref{example1}, observe that in the 2-hop knowledge case (execution of 2-LocWA), a node has all the information that it would have in the 1-hop knowledge case. Therefore, it can utilize the information to update its state value in a manner that 1-LocWA does, in order to guarantee faster convergence time. As a result, the modified algorithm would always be as fast, in terms of real time as 1-LocWA. Next, we modify the update condition of Algorithm $k$-LocWA to capture this strengthened version with respect to real time.\smallskip 

\noindent \textbf{Update Condition of Strong $k$-LocWA}~ In the strong version of Algorithm $k$-LocWA, a node updates its value the first time that at least one of conditions $i$-WAIT, for $i\in \{1,\ldots, k\}$ holds. Specifically we replace the update condition of Algorithm $k$-LocWA with:
\begin{itemize}[noitemsep,topsep=0pt]	
	\item Update value when   {\em $\displaystyle\bigvee_{i=1}^k (i$-WAIT$)=true$} for the first time in phase $p$:

\end{itemize}

Considering this strong version of the algorithm family $k$-LocWA, we can show that for $k'\ge k$ and any $\epsilon$, Algorithm  $k'$-LocWA will $\epsilon$-converge faster than Algorithm $k$-LocWA. That is, for every delay scenario, the number of rounds in which $k$-LocWA $\epsilon$-converges is larger than the number of rounds in which  $k'$-LocWA $\epsilon$-converges. The proof is trivial, since the strengthened algorithm $k'$-LocWA will check all the update conditions for smaller values of $k$, and the messages communicated in $k'$-LocWA are a superset of the messages communicated in $k$-LocWA. Also observe that if $k$-LocWA $\epsilon$-converges then so does $k'$-LocWA.  Thus we have the following Corollary. 

\begin{corollary}
	For $k'\ge k$, if Strong $k$-LocWA $\epsilon$-converges in $r$ rounds then Strong $k'$-LocWA $\epsilon$-converges in $r'$ rounds with $r'\le r$.  
\end{corollary}

\appendix
\newpage
\section{Additional Discussion of Related Work}
\label{a:related}

\subsection{Consensus}

Lamport, Shostak, and Pease addressed the Byzantine consensus problem in \cite{lamport_agreement}. Subsequent work \cite{impossible_proof_lynch,dolev_82_BG} characterized the necessary and sufficient conditions under which Byzantine consensus is solvable in {\em undirected} graphs. However, these conditions are not adequate to fully characterize the {\em directed} graphs in which Byzantine consensus is feasible. 

Bansal et al. \cite{Bansal_disc11} identified tight conditions for achieving Byzantine consensus in {\em undirected} graphs using {\em authentication}. Bansal et al. discovered that all-pair reliable communication is not necessary to achieve consensus when using authentication. Our work differs from Bansal et al. in that our results apply in the absence of authentication or any other security primitives; also our results apply to {\em directed} graphs. Alchieri et al. \cite{BFT-CUP_OPODIS} explored the problem of achieving exact consensus in {\em unknown} networks with Byzantine nodes, but the underlying communication graph is assumed to be \textit{fully-connected}. In our work, each node has partial network knowledge, and we consider incomplete directed graphs.

\subsection{Iterative Approximate Consensus}

Many researchers in the decentralized control area, including Bertsekas and Tsitsiklis \cite{AA_convergence_markov} and Jadbabaei, Lin and
Morse \cite{jadbabaie_concensus}, have explored approximate consensus in the absence
of faults, using only near-neighbor communication in systems wherein the communication graph may be partially connected and time-varying. Our work considers the case when nodes may suffer crash failures.

Our prior work \cite{vaidya_PODC12,Tseng_general,lili_multihop} has considered a restricted class of iterative algorithms for achieving {\em approximate} Byzantine consensus in directed graphs, where fault-free nodes must agree on values that are approximately equal to each other using iterative algorithms with limited memory (in particular, the state carried by the nodes across iterations must be in the convex
hull of inputs of the fault-free nodes, which precludes mechanisms such as multi-hop forwarding of messages). The conditions developed in such prior work are {\em not} necessary when no such restrictions are imposed. Independently, LeBlanc et al. \cite{leblanc_HiCoNs,Sundaram_journal}, and Zhang and Sundaram \cite{Sundaram,Sundaram_ACC} have developed results for iterative
algorithms for approximate consensus under a {\em weaker} fault model, where a faulty node must send
{\em identical} messages to all the neighbors. 

\subsection{$k$-set Consensus}

$k$-set consensus also received a lot of attentions in different graph assumptions. In complete graphs, Biely et al. \cite{k-set_impossible}  presented impossibility results of $k$-set consensus in various message passing systems. Guerraoui and Pochon \cite{k_set_topology} studied early-deciding $k$-set agreement using algebraic topology techniques. Our work studies \textit{directed incomplete} graphs.
In synchronous dynamic networks, Biely et al. \cite{k-set_dynamic,k-set_dynamic_NETYS} considered $k$-set consensus with fault-free nodes. Winkler et al. \cite{dynamic_consensus_stable} solved exact consensus in synchronous dynamic networks with unreliable links. The main contribution in \cite{dynamic_consensus_stable} was to identify the shortest period of stability that makes consensus feasible. 
In unknown and dynamic systems, Jeanneau et al. \cite{k-set_FD} relied on failure detectors to solve $k$-set consensus. 
These works only studied synchronous systems, whereas we consider \textit{exact} and \textit{approximate} crash-tolerant consensus in asynchronous systems. Moreover, we do not assume the existence of failure detectors.

\subsection{Reliable Communication and Broadcast}

Several papers have also addressed communication between a single source-receiver pair.
Dolev et al. \cite{Dolev90perfectlysecure} studied the problem of secure communication, which achieves both fault-tolerance and perfect secrecy between a single source-receiver pair in undirected graphs, in the presence of node and link failures. Desmedt and Wang considered the same problem in directed graphs \cite{yvo_eurocrypt02}. Shankar et al. \cite{Shankar_SODA08} investigated reliable communication between a source-receiver pair in directed graphs allowing for an arbitrarily small error probability in the presence of a Byzantine failures. Maurer et al. explored the problem in directed dynamic graphs \cite{reliable_comm_dynamic}. In our work, we do not consider secrecy, and address the {\em consensus} problem rather than the single source-receiver pair problem. Moreover, our work addresses both deterministically correct and randomized algorithms for consensus.

There has also been work \cite{CPA_DISC, Tseng_CPA} on the problem of achieving \textit{reliable broadcast} with a fault-free source in the presence of local Byzantine faults, which proved {\em tight} condition on the underlying graphs. In this paper, we consider
{\em consensus} problem instead of reliable broadcast problem; furthermore, we allow any node to be faulty.

\section{Necessity of Condition $1$-CCA}
\label{app_1-CCA}

The necessity proof is similar to the necessity proof of Condition CCA in \cite{Tseng_podc2015}.

\begin{theorem}
	\label{Theorem:necessity1hop}
	If graph $G(\sv,\se)$ does not satisfy Condition $1$-CCA, then no iterative one-hop algorithm can achieve asynchronous approximate consensus in $G(\sv,\se)$.
\end{theorem}

\begin{proof}
	The proof is by contradiction. Suppose that there exists an iterative one-hop algorithm $\sa$ which achieves asynchronous approximate consensus in $G(\sv,\se)$, and $G(\sv,\se)$ does not satisfy Condition $1$-CCA. That is, there exists a node partition $L,C,R$ such that $L,R$ are non-empty,  $L \cup C \not\rightarrow R$ and $R \cup C \not\rightarrow L$.
	
	Let $O(L)$ denote the set of nodes $C\cup R$ that have outgoing links to nodes in $L$, i.e., $O(L)=\{i\mid i\in C\cup R, N_i^+\cap L\neq \emptyset\}$. Similarly define $O(R)=\{i\mid i\in C\cup L, N_i^+\cap R\neq \emptyset\}$. Since $L \cup C \not\rightarrow R$ and $R \cup C \not\rightarrow L$, we have that for every $i\in L$, $N_i^-\cap O(L)\le f$ and for every $i\in R$, $N_i^-\cap O(R)\le f$.
	
	Consider a scenario where (i) each node in $L$ has input 0; (ii) each node in $R$ has input $\epsilon$; (iii) nodes in $C$ (if non-
	empty) have arbitrary inputs in $[0,\epsilon]$; (iv) no node crashes;
	and (v) the message delay for communications channels from
	$O(L)$ to $L$ and from $O(R)$ to $R$ is arbitrarily large compared to all the other channels.
	
	
	Consider nodes in $L$. Since messages from the set $O(L)$ take arbitrarily long to arrive at the nodes in $L$, and for every $i\in L$, $N_i^-\cap O(L)\le f$, from the perspective of node $i$, its incoming neighbors in $O(L)$ appear to
	have crashed. The latter yields from the fact that algorithm $\sa$ is one-hop, i.e., the case that for every $i,j\in L$, $N_i^-\cap O(L)= N_i^-\cap O(L)\le f$ can not be excluded by the messages exchanged in $L$ and thus there is a case where all their neighbors in $O(L)$ are crashed. 
	Thus, nodes in $L$ must decide on their output  without waiting to hear from the nodes in $O(L)$. Consequently, to satisfy the validity property, the output at each  node in $L$ has to be 0, since 0 is the input of all the nodes in
	$L$. Similarly, nodes in $R$ must decide their output without  hearing from the nodes in $O(R)$; they must choose output as $\epsilon$, because the input at all the nodes in $R$ is $\epsilon$. Thus, the $\epsilon$-agreement property is violated, since the difference between  outputs at fault-free nodes is not $<\epsilon$. This is a contradiction.
\end{proof}

\section{Sufficiency of Condition $1$-CCA}
\label{app_1-CCA2}

We first prove a useful lemma.

\begin{lemma}
	\label{lemma:absorb_condition}
	Assume that $G(\scriptv,\scripte)$ satisfies Condition $1$-CCA.
	Consider a partition $A,B$ of $\scriptv$ such that
	$A$ and $B$ are non-empty.
	If $B \not\rightarrow A$, then set $A$ propagates to set $B$.
\end{lemma}

\begin{proof}
	Since $A,B$ are non-empty, and $B\not\rightarrow A$, we have that $A\rightarrow B$ holds, by setting $C=\emptyset$ in Condition $1$-CCA. 
	
	Define $A_0=A$ and $B_0=B$.
	Now, for a suitable $l>0$, we will build propagating sequences $A_0,A_1,\cdots A_l$
	and $B_0,B_1,\cdots B_l$ inductively.
	\begin{itemize}
		\item Recall that $A=A_0$ and $B=B_0\neq \emptyset$. Since $A\rightarrow B$,
		$in(A_0\rightarrow B_0)\neq \emptyset$.
		Define $A_1=A_0\cup in(A_0\rightarrow B_0)$
		and 
		$B_1=B_0-in(A_0\rightarrow B_0)$.
		
		If $B_1=\emptyset$, then $l=1$, and we have found the propagating sequence
		already.
		
		If $B_1\neq \emptyset$, then define $L=A=A_0$, $R=B_1$ and $C=A_1-A=B-B_1$.
		Since $B\not\rightarrow A$, $R\cup C\not\rightarrow L$. Therefore,
		Condition $1$-CCA implies that $L\cup C\rightarrow R$. That is, $A_1\rightarrow B_1$.

		\item 
		For increasing values of $i\geq 0$,
		given $A_i$ and $B_i$, where $B_i\neq\emptyset$, by following steps similar to the previous
		item, we can obtain
		$A_{i+1}=A_0\cup in(A_i\rightarrow B_i)$
		and 
		$B_{i+1}=B_i-in(A_i\rightarrow B_i)$,
		such that either $B_{i+1}=\emptyset$ or $A_{i+1}\rightarrow B_{i+1}$.
	\end{itemize}
	In the above construction, $l$ is the smallest index such that
	$B_l=\emptyset$.
\end{proof}

\paragraph{Proof of Lemma \ref{lemma:must_absorb}}

\begin{proof}
	Consider two cases:
	\begin{itemize}
		\item $A\not\rightarrow B$: Then by Lemma \ref{lemma:absorb_condition} above,
		$B$ propagates to $A$, completing the proof.
		
		\item $A\rightarrow B$: In this case, consider two sub-cases:
		\begin{itemize}
			\item {\em $A$ propagates to $B$}: The proof in this case is complete.
			
			\item {\em $A$ does not propagate to $B$}:
			Recall that $A\rightarrow B$. Since $A$ does not propagate to $B$,
			propagating sequences defined in Definition~\ref{def:absorb_sequence}
			do not exist in this case. More precisely, there must exist $k>0$,
			and sets $A_0,A_1,\cdots,A_k$ and $B_0,B_1,\cdots,B_k$,
			such that:
			\begin{itemize}
				\item $A_0=A$ and $B_0=B$, and
				\item for $0\leq i\leq k-1$,
				\begin{list}{}{}
					\item[o] $A_i\rightarrow B_i$,
					\item[o] $A_{i+1} = A_i\cup in(A_i\rightarrow B_i)$, and
					\item[o] $B_{i+1} = B_i - in(A_i\rightarrow B_i)$.
				\end{list}
				\item $B_{k}\neq \emptyset$ \underline{and} $A_{k}\not\rightarrow B_{k}$.
			\end{itemize}
			The last condition above violates the requirements for $A$ to propagate
			to $B$.
			
			Now, $A_{k}\neq \emptyset$, $B_k\neq \emptyset$, and $A_k,B_k$ form
			a partition of $\scriptv$. Since $A_{k}\not\rightarrow B_{k}$,
			by Lemma \ref{lemma:absorb_condition} above,
			$B_k$ propagates to $A_k$.

			Given that $B_k\subseteq B_0 = B$, $A=A_0\subseteq A_k$, and $B_k$ propagates
			to $A_k$, now we prove that $B$ propagates to $A$. 
			
			Recall that $A_i$ and $B_i$ form a partition of $\scriptv$.
			
			Let us define $P=P_0=B_k$ and $Q=Q_0=A_k$. Thus, $P$ propagates to $Q$.
			Suppose that $P_0,P_1,...P_m$ and $Q_0,Q_1,\cdots,Q_m$ are
			the propagating sequences in this case, with $P_i$ and $Q_i$ forming
			a partition of $P\cup Q = A_k\cup B_k=\scriptv$. \\
			
			Let us define $R=R_0=B$ and $S=S_0=A$.
			Note that $R,S$ form a partition of $A\cup B=\scriptv$.
			Now, $P_0=B_k\subseteq B=R_0$ and $S_0=A\subseteq A_k =Q_0$.
			Also, $R_0-P_0$ and $S_0$ form a partition of $Q_0$.
			Figure~\ref{f_appendixB} illustrates some of the sets used in this
			proof. \\
			
			\begin{figure*}[tbh]
				\centering
				\includegraphics[width=0.6\textwidth]{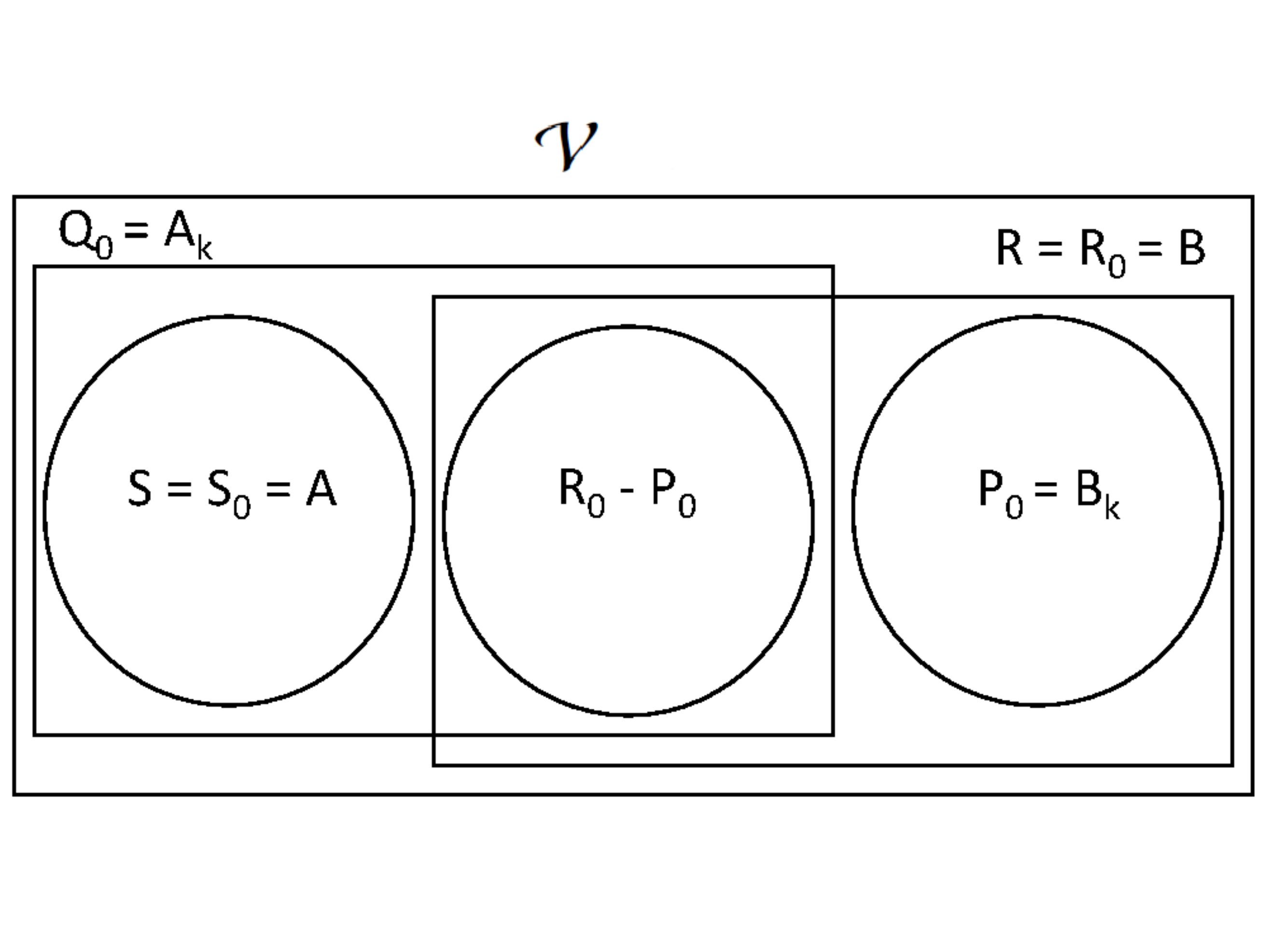}
				\caption{Illustration for the last part of the proof of Lemma \ref{lemma:must_absorb}. In this figure, $R_0 = P_0 \cup (R_0 - P_0)$ and $Q_0 = S_0 \cup (R_0 - P_0)$.}
				\label{f_appendixB}
			\end{figure*}
			
			\begin{itemize}
				\item
				Define $P_1 = P_0\cup (in(P_0\rightarrow Q_0))$, and $Q_1 = \scriptv- P_1 = Q_0 - (in(P_0\rightarrow Q_0))$. Also,
				$R_1 = R_0\cup (in(R_0\rightarrow S_0))$, and $S_1 = \scriptv- R_1 = S_0 - (in(R_0\rightarrow S_0))$.
				
				Since $R_0-P_0$ and $S_0$ are a partition of $Q_0$,
				the nodes in $in(P_0\rightarrow Q_0)$ belong to one of these
				two sets. Note that $R_0-P_0\subseteq R_0$.
				Also, $S_0 \cap in(P_0\rightarrow Q_0) \subseteq in(R_0\rightarrow S_0)$.
				Therefore, it follows that $P_1 = P_0\cup (in(P_0\rightarrow Q_0))
				\subseteq R_0\cup (in(R_0\rightarrow S_0)) = R_1$.
				
				Thus, we have shown that, $P_1\subseteq R_1$. Then it follows that
				$S_1\subseteq Q_1$.
				\\
				\item For $0\leq i<m$, let us define $R_{i+1}=R_i\cup in(R_i\rightarrow 
				S_i)$ and $S_{i+1} = S_i - in(R_i\rightarrow S_i)$. Then following an
				argument similar to the above case, we can inductively show that,
				$P_i\subseteq R_i$ and $S_i\subseteq Q_i$.
				Due to the assumption on the length of the propagating
				sequence above, $P_m=P\cup Q=\scriptv$ and $Q_m=\emptyset$.
				Thus, there must exist $r\leq m$, such that for $i<r$, $R_i\neq \scriptv$, and $R_r=\scriptv$
				and $S_r=\emptyset$.
				
				The sequences $R_0,R_1,\cdots,R_r$ and $S_0,S_1,\cdots,S_r$ form
				propagating sequences, proving that $R=B$ propagates to $S=A$. 
			\end{itemize}
		\end{itemize}
	\end{itemize}
\end{proof}

\red{
\subsection{Proof of Lemma~\ref{lemma:bounded_value}}
~

We first present two additional lemmas (using
the notation in Algorithm LocWA). We will use the notation for $\alpha_i=\frac{1}{|N_i^-|}$ for convenience. Note that $heard_i^*[p], R_i^*[p]$ represents the $heard_i[p], R_i[p]$ sets of node $i$ in phase $p$ the first time that condition 1-WAIT is satisfied. 

\begin{lemma}
\label{lemma:psi}

For node $i\in\scriptv-F[p]$.
Let $\psi\leq \mu[p-1]$. Then, for $j\in heard^*[p]$,
 \[
v_i[p] - \psi  ~\geq~  a_i ~ (v_j[p-1] - \psi)
\]
\end{lemma}
\begin{proof}
In Algorithm LocWA, for each $j\in heard^*[p]$, it holds by definition $\mu[p-1] \leq v_j[p-1]$.
Therefore, 
\begin{eqnarray}
v_j[p-1]-\psi\geq 0 \mbox{\normalfont~for all~} j\in heard_i^*[p]
\label{e_algo_1}
\end{eqnarray}
Since weights in (\ref{mean1}) in Algorithm 1 add to 1, we can re-write that equation
as,
\begin{eqnarray}
v_i[t] - \psi &=& \sum_{j\in heard^*[t]} \frac{1}{R_i^*[p]} \, (v_j[p-1]-\psi) \\
\nonumber
&\geq& \frac{1}{R_i^*[p]}\, (v_j[p-1]-\psi), ~~\forall j\in heard^*[p]  ~~~~~\mbox{\normalfont from (\ref{e_algo_1})}\\
&\geq& \alpha_i \, (v_j[p-1]-\psi), ~~\forall j\in heard^*[p]~~~~~~~~~~\mbox{\normalfont by definition of $\alpha_i$}  
\end{eqnarray}

\end{proof}

~

\begin{lemma}
\label{lemma:Psi}
For node $i\in\scriptv-F[p]$, let $\Psi\geq U[p-1]$. Then, for $j\in heard^*[p]$,

\[
\Psi - v_i[p] \geq  a_i ~ (\Psi - v_j[p-1])
\]
\end{lemma}

\begin{proof}
The proof is similar to Lemma \ref{lemma:psi} proof.
\end{proof}

Next we present the main lemma used in proof of convergence.

\paragraph{\bf\large Proof of Lemma 2}
~

\begin{proof}
Since $R$ propagates to $L$, as 
per Definition~\ref{def:absorb_sequence},
there exist sequences of sets
$R_0,R_1,\cdots,R_l$ and $L_0,L_1,\cdots,L_l$, where
\begin{itemize}
\item $R_0=R$, ~~ $L_0=L$, ~~ $R_l=R\cup L$, ~~ $L_l=\emptyset$, ~~ for $0\leq \tau<l$, $L_\tau \neq \emptyset$, and
\item for $0\leq \tau\leq l-1$,
\begin{list}{}{}
\item[*] $R_\tau\rightarrow L_\tau$,
\item[*] $R_{\tau+1} = R_\tau\cup in(R_\tau\rightarrow L_\tau)$, and
\item[*] $L_{\tau+1} = L_\tau - in(R_\tau\rightarrow L_\tau)$
\end{list}
\end{itemize}
Let us define the following bounds on the states of the fault-free nodes
in $R-F[p]$ at the end of the $p$-th phase:
\begin{eqnarray}
M & = & max_{j\in R-F[p]}~ v_j[p] \\ \label{e_M}
m & = & min_{j\in R-F[p]}~ v_j[p] \label{e_m}
\end{eqnarray}
By the assumption in the statement of Lemma~\ref{lemma:bounded_value},
\begin{eqnarray}
M-m\leq \frac{U[s]-\mu[s]}{2} \label{e_M_m}
\end{eqnarray}
Also, $M\leq U[s]$ and $m\geq \mu[s]$.
Therefore, $U[s]-M\geq 0$ and $m-\mu[s]\geq 0$.

The remaining proof of Lemma~\ref{lemma:bounded_value} relies
on derivation of the three intermediate claims below. \\

\begin{claim}
\label{claim:1}
For $0\leq \tau\leq l$, for each node $i\in R_\tau-F[p+\tau]$,
\begin{eqnarray}
v_i[p+\tau] - \mu[p] ~ \geq~   \alpha^{\tau}(m-\mu[p])
\label{e_ind_1}
\end{eqnarray}
\end{claim}
\noindent{\em Proof of Claim \ref{claim:1}:}
The proof is by induction.

\noindent
{\em Induction basis:}
By definition of $m$, (\ref{e_ind_1}) holds true
for $\tau=0$.

\noindent{\em Induction:}
Assume that (\ref{e_ind_1}) holds true for some $\tau$, $0\leq \tau<l$.
Consider $R_{\tau+1}$.
Observe that $R_\tau$ and $R_{\tau+1}-R_\tau$ form a partition of $R_{\tau+1}$;
let us consider each of these sets separately.
\begin{itemize}
\item Set $R_\tau$: By assumption, for each $i\in R_\tau - F[p+\tau+1]$, (\ref{e_ind_1})
holds true.
By validity of Algorithm LocWA~\footnote{Validity is trivially true due to how Algorithm LocWA updates each node's state.}, $\mu[p] \leq \mu[p+\tau]$.
Therefore, setting $\psi=\mu[p]$ and $t=p+\tau+1$ in Lemma~\ref{lemma:psi}, 
we get,
\begin{eqnarray*}
v_i[p+\tau+1] - \mu[p] & \geq 
& a_i~(v_i[p+\tau] - \mu[p]) \\
& \geq & a_i~ \alpha^{\tau}(m-\mu[s]) ~~~~ \mbox{due to (\ref{e_ind_1})} \\
& \geq & \alpha^{\tau+1}(m-\mu[s])  ~~~~\mbox{due to the definition of $\alpha_i$} \\
	&& \mbox{~~ and because~~~~} m-\mu[s]\geq 0 
\end{eqnarray*}

\item Set $R_{\tau+1}-R_\tau$: Consider a node $i\in R_{\tau+1}-R_\tau- F[p+\tau+1]$. By definition
of $R_{\tau+1}$, we have that $i\in in(R_\tau\rightarrow L_\tau)$.
Thus,
\[ |N_i^- \cap R_\tau| \geq f+1 \] 
%

Since there are at most $f$ faults and $|N_i^- \cap R_\tau| \geq f+1 $, there will exist a node $w\in N_i^- \cap R_\tau \cap heard^*[p+\tau+1] $.
 Then, by an argument similar to the
previous case, we can set $\psi=\mu[s]$
in Lemma~\ref{lemma:psi}, to obtain,
\begin{eqnarray*}
v_i[p+\tau+1] -\mu[s] & \geq & a_i~(v_w[p+\tau]-\mu[p]) \\
& \geq & a_i~ \alpha^{\tau}(m-\mu[p]) ~~~~ \mbox{due to (\ref{e_ind_1})} \\
& \geq & \alpha^{\tau+1}(m-\mu[p])  ~~~~\mbox{due to the definition of $\alpha_i$}\\
	&& \mbox{and because~~~~} m-\mu[s]\geq 0 
\end{eqnarray*}

\end{itemize}

Thus, we have shown that for all nodes in $R_{\tau+1}$,
\[
v_i[s+\tau+1] -\mu[s] 
\geq \alpha^{\tau+1}(m-\mu[s]) 
\]
This completes the proof of Claim \ref{claim:1}. \\

\begin{claim}
\label{claim:2}
For each node $i\in \scriptv-F[p+l]$,
\begin{eqnarray}
v_i[p+l] - \mu[p] ~ \geq ~  \alpha^{l}(m-\mu[p])
\label{e_ind_2}
\end{eqnarray}
\end{claim}
\noindent{\em Proof of Claim \ref{claim:2}:}
Note that by definition, $R_l = \scriptv$. Then the proof follows by setting $\tau = l$ in the above Claim \ref{claim:1}. \\

\begin{claim}
\label{claim:3}
For each node $i\in \scriptv-F[p+l]$,
\begin{eqnarray}
U[p] - v_i[p+l] \geq  \alpha^{l}(U[p]-M)
\label{e_ind_3a}
\end{eqnarray}
\end{claim}

The proof of Claim \ref{claim:3} is similar to the proof of Claim \ref{claim:2}. \\


\noindent
Now let us resume the proof of the Lemma \ref{lemma:bounded_value}.
Note that $R_l=\scriptv$. Thus, 
\begin{eqnarray}
U[p+l] & = & \max_{i\in\scriptv-F[p+l]}~ v_i[p+l] \nonumber \\
& \leq & U[s] - \alpha^{l}(U[p]-M) \mbox{~~~~~~~by (\ref{e_ind_3a})}
\label{e_U}
\end{eqnarray}
and
\begin{eqnarray}
\mu[p+l] & = & \min_{i\in\scriptv-F[p+l]}~ v_i[p+l] \nonumber \\
& \geq & \mu[p] + \alpha^{l}(m-\mu[p]) \mbox{~~~~~~~by (\ref{e_ind_2}})
\label{e_mu}
\end{eqnarray}
Subtracting (\ref{e_mu}) from (\ref{e_U}),
\begin{eqnarray}
&& U[p+l]-\mu[p+l] \nonumber \\  & \leq & U[p] - \alpha^{l}(U[p]-M)  - \mu[p] - \alpha^{l}(m-\mu[p]) \nonumber \\
&=& (1-\alpha^l)(U[p]-\mu[p]) + \alpha^l(M-m) \nonumber \\
&\leq& (1-\alpha^l)(U[p]-\mu[p]) + \alpha^l~\frac{U[p]-\mu[p]}{2} \nonumber
 \mbox{~~~~by (\ref{e_M_m})} \nonumber \\
&\leq& (1-\frac{\alpha^l}{2})(U[p]-\mu[p])  \nonumber
\end{eqnarray}
This concludes the proof of Lemma~\ref{lemma:bounded_value}.
\end{proof}

Now, we are ready to present the main proof of Theorem \ref{Theorem:SufficiencyLocWA}.

\subsection{Proof of Theorem~\ref{Theorem:SufficiencyLocWA}} 
~


}

\begin{proof}
	Validity is trivially true due to how Algorithm LocWA updates each node's state. We will prove that, given any $\epsilon>0$, there
	exists $\tau$ such that
	\begin{equation}
	U[t]-\mu[t] \leq \epsilon ~~~\forall t\geq \tau
	\end{equation}
	
	Consider $p$-th phase, for some $p\geq 0$.
	If $U[p]-\mu[p]=0$, then the algorithm has already converged, and the proof
	is complete, with $\tau=p$.
	
	Now consider the case when $U[p]-\mu[p]>0$.
	Partition $\scriptv$ into two subsets, $A$ and $B$, such
	that, for each fault-free node $i\in A$, 
	$v_i[p]\in \left[\mu[p], \frac{U[p]+\mu[p]}{2}\right)$, and
	for each fault-free node $j\in B$,
	$v_j[p] \in \left[\frac{U[p]+\mu[p]}{2}, U[p]\right]$.
	By definition of $\mu[p]$ and $U[p]$, there exist fault-free nodes
	$i$ and $j$ such that $v_i[p]=\mu[p]$ and $v_j[p]=U[p]$.
	Thus, sets $A$ and $B$ are both non-empty.
	By Lemma \ref{lemma:must_absorb}, one of the following two conditions
	must be true:
	\begin{itemize}
		\item Set $A$ propagates to set $B$. Then, define $L=B$ and $R=A$.
		The states of all the fault-free nodes in $R=A$ are confined within an
		interval of length
		$<\frac{U[p]+\mu[p]}{2} - \mu[p] \leq \frac{U[p]-\mu[p]}{2}$.
		
		\item Set $B$ propagates to set $A$. Then, define $L=A$ and $R=B$.
		In this case, states of all the fault-free nodes in $R=B$ are confined within an interval of length
		$\leq U[p]-\frac{U[p]+\mu[p]}{2} \leq \frac{U[p]-\mu[p]}{2}$. 
		
	\end{itemize}
	In both cases above, we have found non-empty sets $L$ and $R$
	such that (i) $L,R$ is a partition of $\scriptv$,
	(ii) $R$ propagates to $L$, and (iii) the states of all fault-free nodes in $R$ are confined
	to an interval of length $\leq \frac{U[p]-\mu[p]}{2}$.
	Suppose that $R$ propagates to $L$ in $l(p)$ steps, where $l(p)\geq 1$.
	Then by Lemma~\ref{lemma:bounded_value},
	\begin{eqnarray}
	U[p+l(p)]-\mu[p+l(p)] \leq \left( 1-\frac{\alpha^{l(p)}}{2}\right)(U[p] - \mu[p])
	\label{e_t}
	\end{eqnarray}
	
	Observe that $\alpha>0$ (defined in Lemma~\ref{lemma:bounded_value}), else Condition $1$-CCA is violated. Then, $n-f-1 \geq l(p)\geq 1$ and $0<\alpha\leq 1$; hence, $0\leq \left( 1-\frac{\alpha^{l(p)}}{2}\right)<1$.
	
	Let us define the following sequence of phase indices:
	\begin{itemize}
		\item $\tau_0 = 0$,
		\item for $i>0$, $\tau_i = \tau_{i-1} + l(\tau_{i-1})$, where $l(p)$ for any given $p$ was defined above.
	\end{itemize}
	
	If for some $i$, 
	$U[\tau_i]-\mu[\tau_i]=0$, then since the algorithm 
	satisfies the validity condition, we will have $U[t]-\mu[t]=0$
	for all $t\geq \tau_i$, and the proof of convergence is complete.
	
	Now suppose that $U[\tau_i]-\mu[\tau_i]\neq 0$ for the values of $i$ in the
	analysis below.
	By repeated application of the argument leading to (\ref{e_t}), we can prove
	that, for $i\geq 0$,
	
	\begin{eqnarray}
	U[\tau_i]-\mu[\tau_i] \leq \left( \Pi_{j=1}^i\left( 1-\frac{\alpha^{\tau_j-\tau_{j-1}}}{2}\right)\right)~(U[0] - \mu[0])
	\end{eqnarray}
	
	For a given $\epsilon$,
	by choosing a large enough $i$, we can obtain
	\[
	\left(\Pi_{j=1}^i\left( 1-\frac{\alpha^{\tau_j-\tau_{j-1}}}{2}\right)\right)~(U[0] - \mu[0]) \leq \epsilon
	\]
	and, therefore,
	\begin{eqnarray}
	U[\tau_i]-\mu[\tau_i] \leq  \epsilon
	\end{eqnarray}
	For $t\geq \tau_i$, by validity of Algorithm LocWA, it follows that
	\[
	U[t]-\mu[t] \leq
	U[\tau_i]-\mu[\tau_i] \leq  \epsilon
	\]
	This concludes the proof.
\end{proof}

\section{Correctness of Algorithm LWA}
\label{app_topo_dir}

Here, we assume that the graph $G(\sv,\se)$ satisfies Condition CCA. In a given execution of Algorithm LWA, define $F[p]$ as the nodes $i$ that have {\em not} computed value $v_i[p]$ for a fixed phase $p$. In the discussion below, we will drop the phase index $p$ for some notation for brevity.
Results in \cite{Tseng_podc2015} implies that Condition WAIT must hold at some point on the local estimated graph $G^i$, e.g., when node $i$ receives every message.
Since $G^i$ is evolving as node $i$ receives more messages. Suppose Condition WAIT holds on $G^{i*}(\sv^{i*}, \se^{i*})$ for the first time at node $i$.
At that point of time, let $heard^*_i[p], R^*_i[p]$ denote the set $heard_i[p]$ and the corresponding multiset $R_i[p]$.
We prove the following lemma. 

\begin{lemma}
\label{lemma:topo_dir}
Fix a phase $p \geq 1$. For any pair of nodes $i, j \in \sv - F[p]$, $heard^*_i[p] \cap heard^*_j[p] \neq \emptyset$.
\end{lemma}

\begin{proof}
	First observe that by construction, $G^{i*} \subseteq G$, and $heard^*_i[p]$ contains identity of nodes only from $G^{i*}$. Moreover, sets $heard^*_i[p]$ and $heard^*_j[p]$ are defined over (potentially) different estimated graphs at $i$ and $j$, respectively.
	
	By definition, there exist two sets $F_i$ and $F_j$ such that Condition WAIT holds for sets $heard^*_i[p]$ and $F_i$ on $G^{i*}$ at node $i$, and for sets $heard^*_j[p]$ and $F_j$ on $G^{j*}$ at node $j$. 
	In other words, 
	\begin{itemize}
		\item $F_i \subseteq \sv$ and $|F_i| \leq f$, 
		\item $F_j \subseteq \sv$ and $|F_j| \leq f$, 
		\item $reach_i(F_i) \subseteq heard^*_i[p]$, and 
		\item $reach_j(F_j) \subseteq heard^*_j[p]$. 
	\end{itemize}
	
	If $reach_i(F_i) \cap reach_j(F_j) \neq \emptyset$, then the proof is complete, 
	since $reach_i(F_i) \subseteq heard^*_i[p]$ and $reach_j(F_j) \subseteq heard^*_j[p]$. 
	Thus, $heard^*_i[p] \cap heard^*_j[p] \neq \emptyset$. 
	
	Now, consider the case when $reach_i(F_i) \cap reach_j(F_j) =\emptyset$. We will derive a contradiction in this case. 
	Recall that $G^{i*}(\sv^{i*}, \se^{i*})$ is the local estimated graph at node $i$, and $reach_i(F_i)$ is defined as the set of nodes that have directed paths to node $i$ in the subgraph induced by the nodes in $\sv^{i*}-F_i$. 
	
	\begin{claim}
		The set of incoming neighbors of set $reach_i(F_i)$ in $G^{i*}$ is equal to the set of incoming neighbors of set $reach_i(F_i)$ in $G$.
	\end{claim}
	
	\begin{proof}
		The claim follows from the observations that $G^{i*} \subseteq G$ and node $i$ receives a message from each node $k \in reach_i(F_i)$, which contains information of all $k$'s incoming neighbors.
	\end{proof}
	
	This claim implies that in graph $G$, the incoming neighbors of set $reach_i(F_i)$ are contained in set $F_i$. 
	Similarly, in graph $G$, the incoming neighbors of set $reach_j(F_j)$ are contained in set $F_j$.
	
	In graph $G$, we will find subsets of nodes $L, C, R$ that violate Condition CCA.
	Let $L = reach_i(F_i)$, $R = reach_j(F_j)$ and $C = \sv - L - R$. Observe that since $reach_i(F_i) \cap reach_j(F_j) =\emptyset$, $L, C, R$ form a partition of $\sv$. Moreover, $i \in reach_i(F_i)$ and $j \in reach_j(F_j)$; hence, $L = reach_i(F_i)$ and $R = reach_j(F_j)$ are both non-empty.
	Recall that $N^-_L$ is the set of incoming neighbors of set $L$.
	By definition, $N^-_L$ is contained in $R \cup C$.
	Since $L = reach_i(F_i)$, the only nodes that may be in $N^-_L$ are also in $F_i$ as argued above, i.e., $N^-_L\subseteq F_i$. By assumption, $|F_i| \leq f$. Therefore,
	$|N^-_L|\leq f$, which implies that $\notpoint{R\cup C}{L}{f+1}$. Similarly, we can argue that $\notpoint{L\cup C}{R}{f+1}$. These two conditions together show that $G$ violates Condition CCA, a contradiction. Thus, $reach_i(F_i) \cap reach_j(F_j) \neq \emptyset$, which implies that $heard^*_i[p] \cap heard^*_j[p] \neq \emptyset$. This completes the proof.
\end{proof}

Similar to the proofs in \cite{Tseng_podc2015,AA_nancy}, the lemma together with simple algebra, it is easy to show that Algorithm LWA achieves Validity and Convergence.

\section{Algorithm LBC and Correctness}
\label{app_topo_undir}

\noindent{\bf Algorithm LBC}~~The algorithm, presented below, assumes that each node has the knowledge of the network size $n$ and its one-hop neighbors, and the algorithm proceeds in asynchronous phases. The algorithm has two phases: \textit{Learn Phase} and \textit{Consensus Phase}. In the Learn Phase, each node will construct its local knowledge about the whole graph $G^i$, whereas in the Consensus Phase, each node uses the estimated graph $G^i$ and its initial input to solve consensus using existing asynchronous consensus algorithms. 

Given a subgraph $G' \subset G$, we will say node $i$ sends a message $(G', L)$, where the first element contains $G'$, and the second element is the tag denoting the \textit{Learn Phase}.

~

\hrule

\vspace*{2pt}

\noindent {\bf Algorithm LBC} for node $i \in \sv$

\vspace*{4pt}

\hrule

\vspace*{4pt}

\textit{Learn Phase}: 

Initially, $G^i := G_{N_i \rightarrow \{i\}}$ \hspace*{0.5in} // subgraph of one-hop neighbors

Send message $(G^i, L)$ to all the outgoing neighbors

While $G^i$ has strictly less than $n$ nodes: 

~~~~Upon receiving $(G', L)$:

~~~~~~~~$G^i := G^i \cup G'$

~~~~~~~~Send message $(G^i, L)$ to all the outgoing neighbors\vspace{1mm}

\vspace*{3pt}

\textit{Consensus Phase}: 

Solve consensus using existing algorithms based on $G^i$ and $v_i$, the initial input.

\hrule

\vspace*{4pt}

\paragraph{Correctness of Algorithm LBC} It is easy to see the following lemma of Condition CCA.

\begin{lemma}
	\label{lemma:f+1}
	If an undirected graph $G$ satisfies Condition CCA, then $G$ is $(f+1)$-connected.
\end{lemma}

The lemma and the fact that the diameter of $G$ isbounded by $n$ imply the following two lemmas.

\begin{lemma}
	\label{lemma:topo2}	
	If an undirected graph $G$ satisfies Condition CCA, then between any pair of fault-free nodes $i$ and $j$, a message from $i$ will be received by $j$ within $n$ phases.
\end{lemma}

Lemma \ref{lemma:topo2} implies the following lemma.

\begin{lemma}
	\label{lemma:topo3}	
	If an undirected graph $G$ satisfies Condition CCA, then each fault-free node $i$ has $G-F[n] \subseteq G^i$ by the end of the Learn Phase, where $F[n]$ is the set of nodes that crashed by the end of the $n$-th phase, i.e., by the end of the Learn Phase.
\end{lemma}

\begin{theorem}
	\label{theorem:topo}	
	If an undirected graph $G$ satisfies Condition CCA, then Algorithm LBC is correct.
\end{theorem}

\begin{proof}
	Liveness is trivial, since Learn Phase only takes $n$ phases, and the consensus algorithm in the Consensus Phase terminates.
	Now, we show that Algorithm LBC achieves Convergence and Validity. We will use Algorithm WA from \cite{Tseng_podc2015} as the consensus algorithm in the Consensus Phase.

	From Lemma \ref{lemma:topo3}, Algorithm WA in the Consensus Phase reaches consensus even if each node $i$ uses $G^i$ as its view of topology. Observe that the only place to use topology information in Algorithm WA is to check whether Condition WAIT is satisfied or not. Then, if Condition WAIT holds on $G$ after nodes in $F[n]$ crashes, it must also hold on $G-F[n]$. Therefore, Algorithm WA is correct, which implies Algorithm LBC satisfies Convergence and Validity.
\end{proof}

\end{document}